\def\thxPACSs{Funded by Deutsche Forschungsgemeinschaft (DFG, German
	Research Foundation), project PACS (FL~1247/1-1, 522475669).}
\def\HUaffil{Humboldt-Universität zu Berlin,
	Department of Computer Science, Algorithm Engineering Group, Germany}
\def\TForcid{https://orcid.org/0000-0003-2203-4386}

\documentclass[a4paper, USenglish, cleveref, autoref, thm-restate]{oasics-v2021}

\hideOASIcs %

\title{Algorithmics for Safe Bicycle Network Design with Bounded Detours in Rural Areas} %

\titlerunning{Safe Bicycle Network Design with Bounded Detours} %

\author{Till Fluschnik}{\HUaffil}{till.fluschnik@hu-berlin.de}{\TForcid}{\thxPACSs}

\authorrunning{T.\ Fluschnik} %

\Copyright{T.\ Fluschnik} %
\ccsdesc[500]{Theory of computation~Design \& analysis of algorithms}
\ccsdesc[500]{Mathematics of computing~Combinatorial optimization}
\ccsdesc[500]{Social \& professional topics~Sustainability} %

\keywords{%
    Parameterized Algorithms,
    NP-hardness,
    Pairwise weighted spanners,
    Preprocessing,
    Integer linear programming,
    Cut generation
} %

\category{} %

\relatedversion{} %

\supplement{Code and data for the experiments are available at \url{https://github.com/buhtig-tf/Safe-Bicycle-Network-Design-with-Bounded-Detours}.}

\nolinenumbers %

\usepackage{booktabs}

\makeatletter
  \def\mk@font#1#2{%
    \expandafter\def\csname mk#1\endcsname##1{%
      \expandafter\gdef\csname #1##1\endcsname{\ensuremath{#2{##1}}}%
    }%
  }
  \mk@font{cal}{\mathcal} %
  \mk@font{bf}{\mathbf} %
  \mk@font{rm}{\mathrm} %
  \mk@font{bb}{\mathbb} %
  \@tfor\xletter:=A,B,C,D,E,F,G,H,I,J,K,L,M,N,O,P,Q,R,S,T,U,V,W,X,Y,Z\do{%
    \expandafter\mkcal\xletter %
    \expandafter\mkbf\xletter %
    \expandafter\mkrm\xletter %
    \expandafter\mkbb\xletter %
  }
  \@tfor\xletter:=a,b,c,d,e,f,g,h,i,j,k,l,m,n,o,p,q,r,s,t,u,v,w,x,y,z\do{%
    \expandafter\mkbf\xletter %
    \expandafter\mkrm\xletter %
  }
\makeatother

\usepackage{xargs}
\newcommandx{\set}[2][1=A1z]{\ifstrequal{#1}{A1z}{\ensuremath{[#2]}}{\ensuremath{[#1,#2]}}}
\newcommandx{\setL}[2][1=1]{\ensuremath{\{#1,\ldots,#2\}}}

\newcommandx{\mydefenv}[6][2=A,4=A,6=A]{%
  \ifstrequal{#2}{A}{\newtheorem{#1}{#3}}{}
  \ifstrequal{#4}{A}{\crefname{#1}{#3}{#3s}}{\crefname{#1}{#3}{#4}}%
  \ifstrequal{#6}{A}{\Crefname{#1}{#5{.}}{#5s{.}}}{\Crefname{#1}{#5{.}}{#5{.}}}
}

\mydefenv{theorem}[B]{Theorem}{Thm}
\mydefenv{lemma}[B]{Lemma}{Lem}
\mydefenv{proposition}[B]{Proposition}{Pro}
\mydefenv{conjecture}[B]{Conjecture}{Conj}
\mydefenv{corollary}[B]{Corollary}[Corollaries]{Cor}
\mydefenv{observation}[B]{Observation}{Obs}[B]
\mydefenv{fact}{Fact}{Fact}
\theoremstyle{definition}
\mydefenv{problem}{Problem}{Prob}
\mydefenv{definition}[B]{Definition}{Def}
\mydefenv{construction}{Construction}{Constr}
\mydefenv{rrule}{Reduction Rule}{RR}
\mydefenv{xrule}{Rule}{Rule}
\mydefenv{myalgo}{Algorithm}{Algo}
\theoremstyle{remark}
\mydefenv{example}[B]{Example}{Ex}
\mydefenv{remark}[B]{Remark}{Rem}
\mydefenv{idea}{Idea}{Idea}
\mydefenv{intuition}{Intuition}{Intuition}

\newcommandx{\decprob}[6][3=Input,5=Question]{
  \begin{problem}[{#1}]\label{prob:#2}
	\textbf{Given} #4,
	the \textbf{question} is whether #6.
  \end{problem}
}

\newcommand{\Q}{\bbQ}
\newcommand{\Qalo}{\bbQ_{\geq 1}}
\newcommand{\N}{\bbN}
\newcommand{\Nzero}{\bbN_0}

\newcommand{\prob}[1]{\textnormal{\textsc{#1}}}

\newcommand{\etal}{et~al.}

\newcommand{\cocl}[1]{\textrm{#1}}
\newcommand{\NP}{\cocl{NP}}
\newcommand{\coNP}{\cocl{coNP}}
\newcommand{\FPT}{\cocl{FPT}}
\newcommand{\W}[1]{\cocl{W}[#1]}
\newcommand{\ETH}{ETH}

\newcommand{\unlessETH}{unless the \ETH{} breaks}
\newcommand{\SETH}{SETH}

\newcommand{\unlessSETH}{unless the \SETH{} breaks}

\newcommand{\cpoly}{\cocl{poly}}
\newcommand{\NPincoNPslashpoly}{\ensuremath{\NP\subseteq\coNP/\cpoly}}
\newcommand{\unlessPK}{unless \NPincoNPslashpoly}

\newcommand{\wpb}{when parameterized by}

\usepackage{mathtools}
\usepackage{etoolbox}
\newcommand{\cqed}{\hfill$\diamond$}

\newcommand{\cif}{\text{if~}}
\newcommand{\cotw}{\text{otherwise}}
\newcommand{\yes}{\emph{yes}}
\newcommand{\no}{\emph{no}}
\newcommand{\ceq}{\coloneqq}

\newcommand{\eps}{\varepsilon}

\usepackage{tikz}
\usetikzlibrary{calc,positioning}
\usetikzlibrary{arrows.meta,positioning}
\usetikzlibrary{patterns.meta}

\usepackage[textsize=footnotesize,backgroundcolor=green!40!white,linecolor=black!60,obeyFinal,disable]{todonotes}

\usepackage{tabularray}
\UseTblrLibrary{booktabs}

\newcommand{\taban}[3]{[{\scriptsize#1}$_\text{#2}^\text{#3}$]}

\newcommand{\appsymb}{{\large$\star$}}
\newcommand{\appref}[1]{{\hyperref[proof:#1]{\appsymb}}}
\newcommand{\apprefX}[1]{{\hyperref[#1]{\appsymb}}}

\newcommand{\sbnsdTsc}{Safe Bicycle Network with Bounded Detours}
\newcommand{\sbnsdAcr}{SBNBD}

\newcommand{\Es}{\ensuremath{E_s}}
\newcommand{\Eu}{\ensuremath{E_u}}
\newcommand{\cst}{\ensuremath{c}}
\newcommand{\trt}{\ensuremath{\theta}}
\newcommand{\TPset}{\ensuremath{\calP}}

\newcommandx{\impro}[2][1=G]{\ensuremath{#1\langle#2\rangle}}

\usepackage{siunitx}
\newcommand{\vsec}[1]{\SI{#1}{\second}}
\newcommand{\vkm}[1]{\SI{#1}{\kilo\meter}}
\newcommand{\vkmh}[1]{\SI{#1}{\kilo\meter\per\hour}}

\DeclareMathOperator{\dist}{dist}

\DeclareMathOperator{\len}{len}

\DeclareMathOperator{\poly}{poly}

\newcommand{\param}[1]{\mathtt{#1}}
\newcommand{\fes}{\param{fes}}
\newcommand{\fvs}{\param{fvs}}
\newcommand{\tw}{\param{tw}}
\newcommand{\twFi}{\overline{\tw}}
\newcommand{\vcn}{\param{vcn}}

\newcommand{\VG}{\mathtt{V}}
\newcommand{\EG}{\mathtt{E}}

\newcommand{\ntp}{p} %
\newcommand{\nue}{\param{m_u}} %
\newcommand{\nuc}{\param{cc_u}} %
\newcommand{\muc}{\param{mcc_u}} %

\newcommand{\cut}{\partial}

\newcommand{\outarcs}{\operatorname{out}}
\newcommand{\inarcs}{\operatorname{in}}

\newcommand{\algfont}[1]{\textsf{#1}}
\newcommand{\ILP}{\algfont{ILP}}
\newcommand{\ILPcut}[1]{\algfont{ILP}$_{\algfont{twc#1}}$}
\newcommand{\ILPp}{\algfont{ILP}$^{\algfont{pre}}$}
\newcommand{\ILPpcut}[1]{\algfont{ILP}$_{\algfont{twc#1}}^{\algfont{pre}}$}

\newcommand{\plotpath}{plots/}

\newcommand{\tikzpreamble}{%
  \tikzstyle{xnode}=[circle,fill,scale=0.45,draw]

  \tikzstyle{xtypeB}=[color=orange]
  \tikzstyle{xtypeA}=[color=blue]
  \tikzstyle{xtypeC}=[color=magenta]
  \tikzstyle{xtypeD}=[color=teal]

  \tikzstyle{xnmarkA}=[rectangle,scale=0.9,draw,xtypeA,ultra thick]
  \tikzstyle{xnmarkB}=[rectangle,scale=1.3,draw,xtypeB,ultra thick]
  \tikzstyle{xnmarkC}=[rectangle,scale=1.7,draw,xtypeC,ultra thick]
  \tikzstyle{xnmarkD}=[rectangle,scale=2.2,draw,xtypeD,ultra thick]
  \tikzstyle{xnmarkE}=[rectangle,scale=2.6,draw=green,ultra thick]

  \tikzstyle{xsedge}=[-,color=black]
  \tikzstyle{xuedge}=[thick,-,color=red!75!black]
  \tikzstyle{xsoledge}=[ultra thick,-,color=green!66!black]

	\tikzstyle{xemark}=[midway,inner sep=1pt,fill=white,opacity=0.8,font=\scriptsize,sloped]

  \tikzstyle{xpath}=[->,>=latex,line width=0.35em,opacity=0.25]

  \tikzstyle{xpathA}=[xpath,xtypeA]
  \tikzstyle{xpathB}=[xpath,xtypeB]
  \tikzstyle{xpathC}=[xpath,xtypeC]
  \tikzstyle{xpathD}=[xpath,xtypeD]
  \tikzstyle{xpathE}=[xpath,color=red]
  \tikzstyle{xpathF}=[xpath,color=brown]

  \tikzstyle{colPNPH}=[fill=red!70]
  \tikzstyle{colXP}=[fill=yellow]
  \tikzstyle{colWH}=[fill=orange]
  \tikzstyle{colXPWH}=[preaction={fill=yellow},,pattern={Lines[%
        angle=90,              %
        distance=6pt,          %
        line width=3pt        %
      ]}, pattern color=orange]
  \tikzstyle{colFPT}=[fill=green]
  \tikzstyle{colFPTPK}=[preaction={fill=green},,pattern={Lines[%
        angle=90,              %
        distance=6pt,          %
        line width=3pt        %
      ]}, pattern color=green!20!white]
  \tikzstyle{colFPTNOPK}=[preaction={fill=green},,pattern={Lines[%
        angle=0,              %
        distance=6pt,          %
        line width=3pt        %
      ]}, pattern color=green!60!black]
}
\newcommand{\xemarkL}[2]{{\color{black}{#1}}|#2}
\newcommandx{\tikzlabel}[3][1=-2,3=2]{\node at (#1*\xr,#3*\yr)[]{(#2)};}

\begin{document}

\maketitle

\begin{abstract}
We introduce the \emph{Safe Bicycle Network with Bounded Detours} (\emph{SBNBD}) problem,
a network-design problem motivated by the upgrade of rural road networks for bicycle traffic. Given an undirected graph whose edges have length and are classified as safe or unsafe
with unsafe edges carrying upgrade costs,
a set of terminal pairs,
an upgrade budget,
and a detour factor $\alpha$,
the task is to upgrade unsafe edges so that every terminal pair is connected by a safe path whose length is at most $\alpha$ times its shortest-path distance in the original network.

We study the computational complexity of SBNBD from a parameterized perspective.
We prove strong NP-hardness even on highly restricted graph classes,
including planar graphs of treewidth two,
graphs with feedback vertex set number one,
and graphs of maximum degree three.
We complement these lower bounds with polynomial-time algorithms for trees and graphs of maximum degree two.
For parameterized complexity,
we show fixed-parameter tractability for the number of unsafe edges and prove matching lower bounds under SETH,
a polynomial-kernel lower bound,
and W-hardness for natural parameters.
Our main positive structural result is an algorithm that maps any instance to an equivalent instance with $O(\mathrm{fes}+p)$ vertices and edges,
where $\mathrm{fes}$ is the feedback edge number and $p$ is the number of terminal pairs;
this yields fixed-parameter tractability for the combined parameter $\mathrm{fes}+p$.

Finally,
we construct and evaluate exact ILP-based algorithms on road networks derived from OpenStreetMap data for small German municipalities and their surrounding rural regions.
The experiments show that these networks have small treewidth upper bounds and moderate feedback edge structure.
We introduce a preprocessing routine based on the reduction algorithm for the parameter~$\mathrm{fes}+p$
and a heuristic cut generation based on tree decompositions.
Both improve exact solving performance,
in particular for harder instances.
Our results for different detour-factor bounds show that a moderate increase in~$\alpha$
can substantially reduce the total length of the upgraded network,
revealing practical trade-offs between upgrade budget and the maximum allowed relative detour.
Our results indicate that structural graph parameters provide a useful algorithmic lens for safe bicycle-network design in rural areas.
\end{abstract}

\newpage
\section{Introduction}
\label{sec:intro}

Improving bicycle networks for urban and rural areas can help to
reduce pollution \cite{MuellerRCNDGGIKN2015},
promote physical activity \cite{OjaTBGKRK2011},
and mitigate transport poverty~\cite{VanDuelmenSK2022}.
Compared with urban areas,
both public transport and bicycle infrastructure are often insufficient in rural regions~\cite{VieroSzell2025}.
Since improving public transport in rural areas is often economically challenging~\cite{OECD1984}
and people use cars even for short-distance rides~\cite{BeckxBDBI2013,GorgesMingardo2025},
cost-effective improvement of bicycle networks becomes a central planning task.
In Germany,
54\% of respondents living in rural areas with less than 20,000 inhabitants
report that they
would like to use bicycles more frequently in the future
\cite{FahrradMonitor2025}.
High-quality, connected cycling infrastructure could
roughly triple cycling use by 2035
(up to~45\% modal share on short trips),
with significant untapped potential also in rural areas \cite{DollBrauerDuffner2024}.
Other mobility solutions like bike-sharing
for rural areas~\cite{PoltimaeRRP2022}
benefit from such infrastructure.
Safety and bounded detour length~\cite{FGSV2010ERA,SchonerLevinson2014}
are key criteria
for a network's quality.

This paper treats the problem of cost-optimal infrastructure upgrades in a given network
to enable safe bicycle trips
for specified terminal pairs with only small detours.
Herein,
the network's edges are classified either as safe or unsafe,
corresponding to low- and high-stress segments~\cite{FurthMN2016}.
Our model is based on the bicycle network improvement problem due to Lim~\etal~\cite{LimDGH22}.
They optimize an aggregate objective over weighted origin-destination pairs,
while we impose hard relative detour constraints.
They study the problem from an OR perspective,
while we study the problem from a computational and parameterized complexity perspective,
with a focus on graph-theoretical properties derived from the applications on rural areas.
Several models similar to ours are already studied in the scientific literature,
yet a fine-grained computational and parameterized complexity analysis seems to be missing to date.
We complement our complexity analysis with experimental algorithmics.

\subsection{Our Model and Decision Problem}

We model the street network as an undirected graph~$G$ with a set~$V$ of vertices (representing junctions or points of interest like schools) and a set~$E\subseteq\binom{V}{2}$ of edges (representing street segments),
where the edge set~$E$ is partitioned into two sets~$\Es$ and~$\Eu$ of \emph{safe} and \emph{unsafe} edges.
Each edge is equipped with a (travel) length $\trt\colon E\to \N$,
representing the length of or the time needed to traverse the street segment.
Moreover,
each unsafe edge is equipped with a cost~$\cst\colon \Eu\to \N$,
representing the cost to make it safe
(e.g, by constructing a bike lane on this street segment).
For an edge subset~$F\subseteq \Eu$,
we denote by~$\impro[G]{F}$ the graph~$G$ with edge set partitioned into~$\Es'=\Es\cup F$ and~$\Eu'=\Eu\setminus F$.

A terminal pair consists of two distinct vertices from~$V$,
representing the endpoints of a requested bicycle trip.
For two distinct vertices~$s,t\in V$,
an $s$-$t$ path~$P$ is a sequence~$(v_0,v_1,\dots,v_\ell)$ of distinct vertices from~$V$ with~$v_0=s$ and~$v_\ell=t$ such that~$\{v_{i-1},v_{i}\}\in E$ for all~$i\in\set{\ell}$.
The length~$\len(P)$ of~$P$ is $\sum_{i=1}^\ell \trt(\{v_{i-1},v_{i}\})$.
With~$\dist_G(s,t)$ we denote the smallest length of an~$s$-$t$ path in~$G$.
We call path~$P$ \emph{safe} if $\{v_{i-1},v_{i}\}\in \Es$ for all~$i\in\set{\ell}$,
and \emph{unsafe} otherwise.
We often drop the graph from the notation when it is clear from the context.
The detour factor
of an $s$-$t$ path~$P$ in~$G$ is~$\len(P)/\dist(s,t) \in \Qalo$.

\decprob{\sbnsdTsc{} (\sbnsdAcr)}{sbnsd}
{a network $G=(V,E,\cst,\trt)$ with $E=\Es\uplus\Eu$,
a set of $\ntp\in\N$ terminal pairs~$\TPset=\{\{s_i,t_i\}\mid i\in\set{\ntp}\}$ from~$V$,
a budget~$k\in\Nzero$,
and a number~$\alpha\in\Qalo$}
{there is a solution~$F\subseteq \Eu$ with~$\cst(F)=\sum_{e\in F}\cst(e)\leq k$ such that
for each~$i\in\set{\ntp}$,
there is a safe $s_i$-$t_i$ path in~$\impro{F}$ of detour factor at most~$\alpha$}

\subsection{Our Contributions}

Our contributions are two-fold.
First,
we provide
a computational and parameterized complexity analysis for \sbnsdAcr{}.
Second,
we study exact ILP-based algorithms experimentally,
focusing on the impact of reduction rules and treewidth-based cut generation.

\subparagraph{Computational and parameterized complexity.}
\cref{fig:hasse} arranges our complexity-theoretic results.
We prove that \sbnsdAcr{} is \NP-hard in very restricted settings,
such as vertex cover number two,
or treewidth two and maximum degree three.
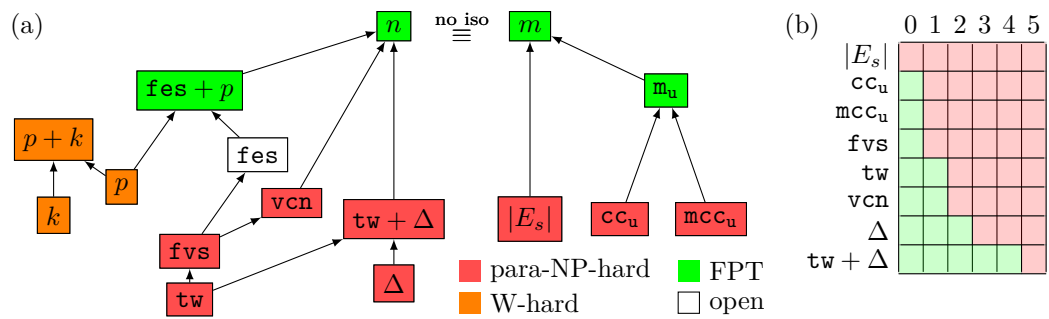
\begin{figure}[t]
 \centering
 \begin{tikzpicture}
  \def\xr{1}
  \def\yr{0.85}
  \def\xs{1.8*\xr}
  \def\ys{1*\yr}
  \def\xss{0.325*\xr}
  \def\yss{0.45*\yr}
  \def\lxsh{2.9*\xr}
  \def\lysh{0.5*\yr}
  \tikzpreamble{}

  \tikzstyle{param}=[draw,rectangle]
  \tikzstyle{paramT}=[draw,rectangle,ultra thick]
  \tikzstyle{pedge}=[->,>=latex]

  \newcommand{\thenodes}{%
    \node (n) at (-1*\xs,0)[param, colFPT]{$n$};
    \node (m) at (0*\xs,0)[param, colFPT]{$m$};
    \node at ($(n)!0.5!(m)$)[]{$\stackrel{\text{no iso}}{\equiv}$};
    \node (ntp) at (-3*\xs,-2.5*\ys)[param,colWH]{$\ntp$};
    \node (fes) at (-2*\xs,-2*\ys)[param,fill=white]{$\fes$};
    \node (fesntp) at (-2.5*\xs,-1*\ys)[param,colFPT]{$\fes+\ntp$};
    \node (twdelta) at (-1*\xs,-3*\ys)[param,colPNPH]{$\tw+\Delta$};
    \node (fvs) at (-2.5*\xs,-3.5*\ys)[param,colPNPH]{$\fvs$};
    \node (vcn) at (-1.75*\xs,-2.75*\ys)[param,colPNPH]{$\vcn$};
    \node (tw) at (-2.5*\xs,-4.25*\ys)[param,colPNPH]{$\tw$};
    \node (delta) at (-1*\xs,-4*\ys)[param,colPNPH]{$\Delta$};
    \node (safe) at (0*\xs,-3*\ys)[param,colPNPH]{$|\Es|$};
    \node (unsafe) at (1*\xs,-1*\ys)[param,colFPT]{$\nue$};
    \node (nuc) at (0.66*\xs,-3*\ys)[param,colPNPH]{$\nuc$};
    \node (muc) at (1.33*\xs,-3*\ys)[param,colPNPH]{$\muc$};
    \node (k) at (-3.5*\xs,-2.5*\xr)[param,colWH]{$k$};
    \node (ntpk) at (-3.5*\xs,-1.5*\xr)[param,colWH]{$\ntp+k$};

  }

  \thenodes{}

  \foreach \x/\y in {fesntp/n,fes/fesntp,ntp/fesntp,twdelta/n,fvs/fes,tw/fvs,tw/twdelta,delta/twdelta,safe/m,unsafe/m,nuc/unsafe,muc/unsafe,ntp/ntpk,k/ntpk,fvs/vcn,vcn/n}{%
    \draw[pedge] (\x) to (\y);
  }

  \thenodes{}

  \begin{scope}[xshift=-0.8*\xr cm, yshift=-3.8*\ys cm]%
	 \node at (0,0*\lysh)[colPNPH,inner sep=4pt,label=0:{para-\NP-hard}]{};
	 \node at (0,-1*\lysh)[colWH,inner sep=4pt,label=0:{W-hard}]{};
	 \node at (1*\lxsh,0*\lysh)[colFPT,inner sep=4pt,label=0:{FPT}]{};
	 \node at (1*\lxsh,-1*\lysh)[param,fill=white,inner sep=4pt,label=0:{open\phantom{A}}]{};
	\end{scope}

	\begin{scope}[xshift=4.875*\xr cm, yshift=-0.25*\ys cm]
    \def\xmax{6}
    \def\ymax{7}
    \foreach\x in {0,1,...,\xmax}{
      \node (v\x) at (\x*\xss,0)[inner sep=0pt]{};
      \node (vv\x) at (\x*\xss,-\ymax*\yss-1*\yss)[inner sep=0pt]{};
      \draw (v\x) -- (vv\x);
      \ifnum\x<6
        \node at (0.5*\xss+\x*\xss,0)[anchor=south]{$\x$};
      \fi
    }
    \foreach\y in {0,1,...,\ymax}{
      \node (p\y) at (0,-\y*\yss)[inner sep=0pt]{};
      \node (pp\y) at (\xmax*\xss,-\y*\yss)[inner sep=0pt]{};
      \draw (p\y) -- (pp\y);
      \node (pl\y) at (0,-0.5*\yss-\y*\yss)[]{};
    }
    \draw (0,-\ymax*\yss-1*\yss) -- (\xmax*\xss,-\ymax*\yss-1*\yss);

    \foreach\y/\l/\ptime/\nphard in {%
      2/{$\muc$}/1/1,1/{$\nuc$}/1/1,
      0/{$|\Es|$}/0/0,4/{$\tw$}/2/2,
      3/{$\fvs$}/1/1,5/{$\vcn$}/2/2,
      6/{$\Delta$}/3/3,7/{$\tw+\Delta$}/5/5%
    }{
      \node at (pl\y)[anchor=east]{\l};
      \draw[fill=green,opacity=0.2] (p\y) rectangle ($(p\y)+(\ptime*\xss,-1*\yss)$);
      \draw[fill=red,opacity=0.2] ($(p\y)+(\nphard*\xss,-1*\yss)$) rectangle (pp\y);
    }
	\end{scope}
  \node at (-6.65*\xr,0*\yr)[]{(a)};
  \node at (3.6*\xr,0*\yr)[]{(b)};
 \end{tikzpicture}
 \caption{(a) Hasse diagram of our studied parameters with their complexity classification assuming unary encoding. Number~$n$ of vertices and number~$m$ of edges are parameterized equivalent when there are no isolated vertices.
 A parameter~$x$ points to another parameter~$y$ if there is a function~$f$ such that~$x\leq f(y)$ holds for every instance.
 Thus,
 fixed-parameter tractability for~$x$ transfers to $y$,
 while hardness for $y$ transfers to $x$.
 (b) Overview of the polynomial-time solvability (green) versus \NP-hardness (red) borders for para-\NP-hard parameters.
 ($\nue$: number of unsafe edges;
 $\nuc$: number of components in the graph induced by all unsafe edges;
 $\muc$: size of the largest component in the graph induced by all unsafe edges.)}
 \label{fig:hasse}
\end{figure}
Each of our para-\NP-hardness results is tight
in the sense that the next smaller parameter value leads to polynomial-time solvability.
We prove \W{1}-hardness \wpb{} number of terminal pairs combined with the budget.
For binary-encoded edge lengths and costs,
we prove \NP-hardness already for one terminal pair.
Finally,
we show several reduction rules
that transform any instance to an equivalent instance whose combinatorial size
is linear in the number of terminal pairs
and the feedback edge number,
i.e.,
the size of a minimum feedback edge set.
These reduction rules
are also effective as preprocessing in our experiments.

\subparagraph*{Experimental algorithmics.}
Using OpenStreetMap data,
we extracted undirected road networks for 30 German villages and their \vkm{3} surrounding rural areas.
These villages were selected from a ranking of over 400 German villages according to their perceived bicycle friendliness.
For these networks,
we then computed the maximum degree,
the size of a minimum feedback edge set,
small upper bounds on the treewidth,
and checked for planarity.
Our main findings are that these networks
have small maximum degree~$\leq 5$,
small treewidth upper bound~$\leq 14$,
and roughly 20\% of the edges
suffice on average
as a feedback edge set.

We tested a plain ILP against
the ILP with preprocessing
via reduction rules,
with additional cuts derived from a tree decomposition (TD),
and with both.
Our results are:
\begin{enumerate}
 \item Most instances obtained for villages are solvable within seconds by each of the solvers.
 \item For the larger instances,
 preprocessing combined with cut generation performed best.
 \item The total length of upgraded unsafe edges decreases by $7.87\%$ on average
 when~$\alpha$ increases from~$1.2$ to~$1.3$,
 illustrating the trade-off between allowed detours and required upgrades.
\end{enumerate}

\section{Related Work}

\subparagraph*{Improving bicycle networks.}
One line of research develops computational optimization models for bicycle-network design and improvement.
Lim~\etal~\cite{LimDGH22}
consider weighted terminal pairs
and
optimize the weighted detour penalty aggregated over each trip.
They formulate their problem as mixed integer program (MIP) and solve it via Benders decomposition on a case-study instance.
Duthie and Unnikrishnan~\cite{DuthieUnnikrishnan2014}
study ``retrofitting'' segments at minimum cost to ensure a minimum safety level with a given bounded detour.
Compared to our model,
their segments can have several safety levels and they express their length-bound through a function of the shortest path.
While structurally,
their focus is on roadway quality,
ours is on parameters observed in rural areas.
Mauttone~\etal~\cite{MauttoneMRT2017}
model user route choices in a multi-commodity-flow MIP solved heuristically,
allowing routes to use segments without bicycle infrastructure at higher cost.
In contrast, we require each terminal pair to have safe route of bounded relative detour.
Natera et al.~\cite{NateraOBIS2020} use a greedy strategy to compute missing links to interconnect components of the bicycle network.
Steinacker~\etal~\cite{SteinackerSTM2022} use greedy network sparsification to design efficient urban bicycle networks,
incorporating cycling demand and route choices based on safety preferences.
To the best of our knowledge,
computational or parameterized complexity analyses are missing for bicycle-network improvement problems with infrastructure upgrades under hard detour constraints.
Another line of research analyzes existing bicycle networks and their growth.
Szell~\etal~\cite{SzellMPGS2022} argue that uncoordinated enhancement of bicycle networks
can lead to higher cost and outline the need for careful planning.
Schoner and Levinson~\cite{SchonerLevinson2014} analyzed 74 bicycle networks from the US
and found that connectivity and directness are important factors.
Finally,
Furth~\etal~\cite{FurthMN2016} developed a framework to classify street segments as low- to high-stress segments.

\subparagraph*{Algorithmics of related problems.}
By assigning zero cost to safe edges and positive cost to unsafe edges,
our problem
can be viewed as
an
undirected pairwise weighted spanner or as a distance-constrained Steiner-type problem.
Cygan~\etal~\cite{CyganGK2013} study sparse pairwise spanners in unweighted undirected graphs
where the possible increase of a pairwise distance is controlled by a stretch function.
They study existential bounds and give a polynomial-time construction for additive stretch functions.
Approximation algorithms and hardness results for spanner-type problems
have been studied on unweighted graphs~\cite{ChlamtacDKL2020}
and weighted directed graphs~\cite{GrigorescuKL2026}.
Kobayashi~\cite{Kobayashi2018}
studies the all-pairs $t$-spanners in unweighted undirected graphs,
where each pair's distance is allowed to increase within a factor of~$t$.
They prove \NP-hardness even in planar bounded-degree graphs.
Moreover,
they show an FPT algorithm when parameterized by the number of discarded edges
(Kobayashi~\cite{Kobayashi2020} proved such an algorithm to exist also for additive spanners).
Regarding parameterized complexity,
Feldmann and Lampis~\cite{FeldmannLampis2025}
study the related problem \prob{Steiner Forest} with terminal pairs
(which has no detour constraints).
They focus on structural parameters
and prove an EPAS with running time depending on treewidth
(the problem is para-\NP-hard already for $\tw=3$~\cite{GASSNER2010}).
They also prove an FPT algorithm when parameterized by the size of a feedback edge set.
Simonov~\cite{SimonovSV2026} study
the problem of computing pairwise distance preservers in undirected unweighted graphs,
where no increase in path length is allowed between any terminal pair.
They prove \W{1}-hardness when parameterized by the number of terminal pairs and
\NP-hardness for vertex cover number $\vcn=3$.
Thus,
we can view our problem as an undirected
common-stretch
pairwise
weighted
spanner problem with designated safe edges of zero cost and unsafe edges of positive cost;
to the best of our knowledge,
a systematic parameterized complexity study is missing
for this variant.

\subparagraph*{Structural parameters of real-world graphs.}
Maniu~\etal~\cite{ManiuSJ19} study
the treewidth of 25 networks from eight different domains,
including infrastructure.
They find that many of these networks have treewidth too large for a direct application of treewidth-based algorithms.
In contrast,
we show that the street networks derived from rural areas admit small treewidth upper bounds.
Analyzing large-scale street networks using OSM data has become common~\cite{Boeing2017,Boeing2020,EppsteinGoodrich2008}.
Related to the feedback edge number,
the meshedness coefficient of planar real-world street networks is studied~\cite{CardilloSLP2006},
which relates the numbers of bounded faces of the network
and of a maximally connected planar graph with the same number of vertices.
In connected planar graphs,
the number of bounded faces equals the feedback edge number.
To the best of our knowledge,
classic FPT parameters such as treewidth or feedback edge number
have not been measured
for larger datasets of rural street networks.

\section{Preliminaries}

We denote by~$\N$ and~$\Nzero$ the natural numbers excluding and including zero,
respectively.
We denote by~$\Q$ and~$\Qalo$ the rational numbers and the rational numbers at least one,
respectively.
We use basic notation from parameterized complexity~\cite{cygan2015parameterized}.
We use common graph-theoretic notation and explain only notation that may be non-standard.
For an undirected graph~$G=(V,E)$,
we denote the vertex and edge set of~$G$ by~$\VG(G)$ and~$\EG(G)$,
respectively.
For an edge subset~$E'\subseteq E$,
we denote by~$G[E']=(V',E')$ with~$V'=\{v\in V\mid \exists e\in E': v\in e\}$
the edge-induced graph~$G$ induced on the edge set~$E'$.
We write~$G-E'$ for the graph~$(V,E\setminus E')$.
For a vertex subset~$U\subseteq V$,
we denote by~$\cut_G(U)=\{\{v,w\}\in E\mid |\{v,w\}\cap U|=1\}$
the \emph{cut} for~$U$ in~$G$,
i.e.,
all edges in~$G$ with exactly one endpoint in~$U$.

\section{Tight Strong NP-hardness on Simple Graph Classes}

\begin{theorem}
 \label{thm:NPh:planarfvs1}
 \sbnsdAcr{} is \NP-hard even on planar graphs with
 treewidth $\tw=2$, and
 \begin{enumerate}[(a)]
  \item feedback vertex set size $\fvs=1$, $\muc=1$, and the graph is bipartite.
  \item feedback vertex set size $\fvs=1$, vertex cover number~$\vcn=2$, and $\nuc=1$.
  \item maximum degree~$\Delta=3$, $\muc=1$, and the graph is bipartite.
 \end{enumerate}
 Moreover,
 \unlessETH,
 there is no $2^{o(n+m)}\cdot \poly(|I|)$-time algorithm.
\end{theorem}

We give three polynomial-time many-one reductions,
each from the \NP-hard \cite{Karp1972} \prob{Vertex Cover} problem,
where,
given an undirected graph~$G=(V,E)$ and an integer~$r$,
the question is whether there is a subset~$W\subseteq V$ with~$|W|\leq r$
such that~$e\cap W\neq \emptyset$ for every~$e\in E$.
\prob{Vertex Cover} admits no~$2^{o(|V|+|E|)}\cdot \poly(|I|)$-time algorithm
\unlessETH~\cite{ImpagliazzoPaturi2001,ImpagliazzoPZ2001}.

\newcommand{\nvc}{\eta}
\begin{construction}
 \label{constr:planarfvs1}
 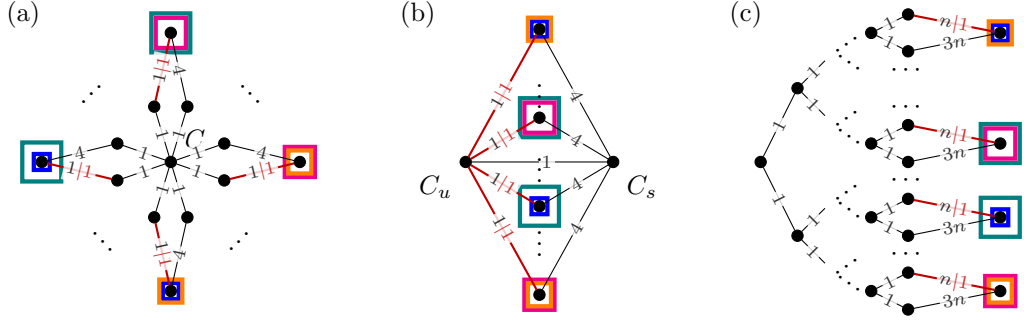
\begin{figure}
  \centering
  \begin{tikzpicture}
   \def\xr{0.975}
   \def\yr{0.975}
   \tikzpreamble{}

   \newcommand{\skelA}{%
      \node (c) at (0,0)[xnode,label=45:$C$]{};

      \node (vu) at (-0.225*\xr,-0.75*\yr)[xnode]{};
      \node (vs) at (+0.225*\xr,-0.75*\yr)[xnode]{};
      \node (va) at (0*\xr,-1.75*\yr)[xnode]{};
      \node at (va)[xnmarkA]{};
      \node at (va)[xnmarkB]{};

      \node (uu) at (-0.225*\xr,0.75*\yr)[xnode]{};
      \node (us) at (+0.225*\xr,0.75*\yr)[xnode]{};
      \node (ua) at (0*\xr,1.75*\yr)[xnode]{};
      \node at (ua)[xnmarkD]{};
      \node at (ua)[xnmarkC]{};

      \node (xu) at (-0.725*\xr,-0.25*\yr)[xnode]{};
      \node (xs) at (-0.725*\xr,0.25*\yr)[xnode]{};
      \node (xa) at (-1.75*\xr,0*\yr)[xnode]{};
      \node at (xa)[xnmarkA]{};
      \node at (xa)[xnmarkD]{};

      \node (yu) at (+0.725*\xr,-0.25*\yr)[xnode]{};
      \node (ys) at (+0.725*\xr,0.25*\yr)[xnode]{};
      \node (ya) at (+1.75*\xr,0*\yr)[xnode]{};
      \node at (ya)[xnmarkC]{};
      \node at (ya)[xnmarkB]{};

      \foreach\x in {v,u,x,y}{
        \draw[xsedge] (\x u) -- node[xemark]{$1$}(c);
        \draw[xuedge] (\x u) -- node[xemark]{$\xemarkL{1}{1}$}(\x a);
        \draw[xsedge] (\x s) -- node[xemark]{$4$}(\x a);
        \draw[xsedge] (\x s) -- node[xemark]{$1$}(c);
      }

      \node at (1*\xr,1*\xr)[rotate=45]{$\vdots$};
      \node at (-1*\xr,1*\xr)[rotate=-45]{$\vdots$};
      \node at (-1*\xr,-1*\xr)[rotate=45]{$\vdots$};
      \node at (1*\xr,-1*\xr)[rotate=-45]{$\vdots$};
   }

   \begin{scope}[xshift=0*\xr cm]
      \tikzlabel{a}
      \skelA{}
   \end{scope}

   \newcommand{\skelB}{%
      \node (Cu) at (-1*\xr,0)[xnode,label=-135:$C_u$]{};
      \node (Cs) at (1*\xr,0)[xnode,label=-45:$C_s$]{};

      \node (va) at (0*\xr,1.8*\yr)[xnode]{};
      \node at (va)[xnmarkA]{};
      \node at (va)[xnmarkB]{};

      \node (ua) at (0*\xr,0.6*\yr)[xnode]{};
      \node at (ua)[xnmarkD]{};
      \node at (ua)[xnmarkC]{};

      \node (xa) at (0*\xr,-0.6*\yr)[xnode]{};
      \node at (xa)[xnmarkA]{};
      \node at (xa)[xnmarkD]{};

      \node (ya) at (0*\xr,-1.8*\yr)[xnode]{};
      \node at (ya)[xnmarkC]{};
      \node at (ya)[xnmarkB]{};

      \foreach\x in {v,u,x,y}{
        \draw[xuedge] (Cu) -- node[xemark]{$\xemarkL{1}{1}$}(\x a);
        \draw[xsedge] (Cs) -- node[xemark]{$4$}(\x a);
      }
      \draw[xsedge] (Cs) -- node[xemark,xshift=0.1*\xr cm]{$1$}(Cu);

      \node at (0*\xr,1*\xr){$\vdots$};
      \node at (0*\xr,0*\xr){$\vdots$};
      \node at (0*\xr,-1*\xr){$\vdots$};
   }

   \begin{scope}[xshift=5*\xr cm]
      \tikzlabel[-1.65]{b}
      \skelB{}
   \end{scope}

   \newcommand{\skelC}{%
      \node (c) at (0,0)[xnode]{};
      \node (ct) at (0.5*\xr,1*\yr)[xnode]{};
      \node (cb) at (0.5*\xr,-1*\yr)[xnode]{};

      \node (ctt) at (1*\xr,1.5*\yr)[]{};
      \node (ctb) at (1*\xr,0.5*\yr)[]{};
      \node (cbt) at (1*\xr,-0.5*\yr)[]{};
      \node (cbb) at (1*\xr,-1.5*\yr)[]{};

      \foreach \x in {ctt,ctb,cbt,cbb}{%
        \node at (\x)[anchor=south west,xshift=-5*\xr pt,yshift=2*\xr pt,rotate=-55]{$\vdots$};
        \node at (\x)[anchor=north west,xshift=-9*\xr pt,yshift=0*\xr pt,rotate=55]{$\vdots$};
      }

      \node (cv) at (1.5*\xr,1.75*\yr)[xnode]{};

      \node (vu) at (2*\xr,2*\yr)[xnode]{};
      \node (vs) at (2*\xr,1.5*\yr)[xnode]{};
      \node (va) at (3.25*\xr,1.75*\yr)[xnode]{};
      \node at (va)[xnmarkA]{};
      \node at (va)[xnmarkB]{};

      \node (cu) at (1.5*\xr,0.25*\yr)[xnode]{};
      \node (uu) at (2*\xr,0.5*\yr)[xnode]{};
      \node (us) at (2*\xr,0*\yr)[xnode]{};
      \node (ua) at (3.25*\xr,0.25*\yr)[xnode]{};
      \node at (ua)[xnmarkD]{};
      \node at (ua)[xnmarkC]{};

      \node (cx) at (1.5*\xr,-0.75*\yr)[xnode]{};
      \node (xu) at (2*\xr,-0.5*\yr)[xnode]{};
      \node (xs) at (2*\xr,-1*\yr)[xnode]{};
      \node (xa) at (3.25*\xr,-0.75*\yr)[xnode]{};
      \node at (xa)[xnmarkA]{};
      \node at (xa)[xnmarkD]{};

      \node (cy) at (1.5*\xr,-1.75*\yr)[xnode]{};
      \node (yu) at (2*\xr,-1.5*\yr)[xnode]{};
      \node (ys) at (2*\xr,-2*\yr)[xnode]{};
      \node (ya) at (3.25*\xr,-1.75*\yr)[xnode]{};
      \node at (ya)[xnmarkC]{};
      \node at (ya)[xnmarkB]{};
      \foreach\x in {v,u,x,y}{ %
        \draw[xsedge] (\x u) -- node[xemark]{$1$}(c\x);
        \draw[xuedge] (\x u) -- node[xemark]{$\xemarkL{n}{1}$}(\x a);
        \draw[xsedge] (\x s) -- node[xemark]{$3n$}(\x a);
        \draw[xsedge] (\x s) -- node[xemark]{$1$}(c\x);
      }

      \foreach\x/\y in {c/ct,c/cb,
        ct/ctt,ct/ctb,cb/cbt,cb/cbb%
        }{
        \draw[xsedge] (\x) -- node[xemark]{$1$}(\y);
      }

      \foreach \x/\y in {2/1.25,2/0.75,2/-0.25,2/-1.25}{%
        \node at (\x*\xr,\y*\xr){$\cdots$};
      }
   }

   \begin{scope}[xshift=8*\xr cm]
      \tikzlabel[-0.2]{c}
      \skelC{}
   \end{scope}
  \end{tikzpicture}
  \caption{Illustrations for \cref{constr:planarfvs1}(a), (b), and (c).
  An unsafe edge~$e$ is drawn red and marked with its length and cost as $\trt(e)\color{red}{|\cst(e)}$.
  Frames depict terminal pairs,
  i.e.,
  two vertices with the same frame (size and color) form a terminal pair.}
  \label{fig:planarfvs1}
 \end{figure}

 Let~$I=(G=(V,E),r)$ be an instance of \prob{Vertex Cover}
 with~$\nvc=|V|$ vertices.
 In each construction,
 we construct a graph~$G'$ that contains a vertex~$v^*$ for each vertex~$v\in V$,
 the set of terminal pairs is~$\TPset=\{\{v^*,w^*\} \mid e=\{v,w\}\in E\}$,
 and we have~$k=r$.
 See \cref{fig:planarfvs1} for an illustration to each of the three constructions.

 (a)
 Add a central vertex~$C$.
 For each vertex~$v\in V$,
 add two vertices~$v_u,v_s$,
 and add unsafe edge~$\{v_u,v^*\}$ of length 1 and cost 1,
 and safe edge~$\{v^*,v_s\}$ of length 4.
 Make~$C$ adjacent to each of~$v_u,v_s$ for each~$v\in V$ via a safe edge of length 1.
 Let~$\alpha\ceq2$.

 (b) There are two central vertices~$C_s,C_u$
 (which form the vertex cover of the constructed graph)
 connected by a safe edge of length~$1$.
 For each vertex~$v\in V$,
 connect~$v^*$ with~$C_s$ via a safe edge of length 4,
 and with~$C_u$ via an unsafe edge of length 1 and cost 1.
 Let~$\alpha\ceq3$.

 (c) We can assume that~$\nvc\ceq |V|=2^q$ for some~$q\in \N$ and~$2\log(\nvc)+4<\nvc$
 (by padding isolated vertices).
 Let~$G'$ be initially a binary tree~$T_\nvc$ with set~$\{\ell_v \mid v\in V\}$ of~$\nvc$ leaves.
 For each vertex~$v\in V$,
 add two vertices~$v_u,v_s$,
 and add unsafe edge~$\{v_u,v^*\}$ of length $\nvc$ and cost~1,
 and safe edge~$\{v^*,v_s\}$ of length $3\nvc$.
 Make~$\ell_v$ adjacent with each of~$v_u,v_s$ for each~$v\in V$ via a safe edge of length 1.
 Let~$\alpha\ceq2$.
 \cqed
\end{construction}

Each constructed network has treewidth at most two.
For (a) and (b),
removing~$C$ and~$C_u$,
respectively,
leaves a forest.
For (c),
start with a tree decomposition of width 1 for~$T_\nvc$.
For every $v\in V$,
choose a node whose bag contains~$\ell_v$
and append a node with bag~$\{\ell_v,v_u,v_s\}$.
To this new node,
append another node with bag~$\{v_u,v_s,v^*\}$.

\begin{proof}[Proof of \cref{thm:NPh:planarfvs1}]
 Let~$I'=(G'=(V',E',\cst,\trt),\TPset,k,\alpha)$ be the instance of \sbnsdAcr{}
 obtained in polynomial time from an instance $I=(G=(V,E),r)$ of \prob{Vertex Cover}
 via \cref{constr:planarfvs1}(a), (b), or (c).
 By construction,
 for each terminal pair~$\{s,t\}$,
 every $s$-$t$ path with detour factor at most~$\alpha$
 must contain an unsafe edge.
 For (a), the shortest path is of length $4$,
 and the shortest safe path is of length $10>8=\alpha\cdot 4$.
 For (b), the shortest path is of length $2$,
 and the shortest safe path is of length $8>6=\alpha\cdot 2$.
 For (c), the shortest path is of length $2\nvc+2+x$,
 where~$x$ is the length of the shortest path through the binary tree,
 and the shortest safe path is of length $6\nvc+2+x>4\nvc+4+2x=\alpha\cdot(2\nvc+2+x)$,
 using $x\leq 2\log(\nvc)$ and our assumption~$2\log(\nvc)+4<\nvc$.
 Consequently,
 for any solution~$F$,
 every feasible safe $s$-$t$ path in $\impro[G']{F}$
 must contain an edge of~$F$.

 Define the following convention for a subset~$F$ of unsafe edges and a set~$W\subseteq V$:
 the unsafe edge incident with~$v^*$ is in~$F$
 if and only if $v\in W$.
 Since in each of the three constructions (a), (b), and (c),
 each unsafe edge has cost one,
 the convention implies that $F$ is of cost at most~$k=r$
 if and only if~$W$ is of size at most~$r$.
 Now,
 for every terminal pair~$\{v^*,w^*\}\in\TPset$,
 let~$\{v,w\}\in E$ be the corresponding edge in the input graph.
 The following equivalences hold.
 There is a safe $v^*$-$w^*$ path of detour factor at most~$\alpha$ in $\impro[G']{F}$
 if and only if at least one of the two unsafe edges incident with~$v^*$ or $w^*$ is in $F$.
 By our convention,
 we have that
 at least one of the two unsafe edges incident with~$v^*$ or $w^*$ is in $F$
 if and only if at least one of~$v,w$ is in~$W$.
 Finally,
 at least one of~$v,w$ is in~$W$
 if and only if the edge~$\{v,w\}\in E$ is covered by~$W$.
 Thus,
 all terminal pairs admit a feasible safe path in~$\impro[G']{F}$
 if and only if~$W$ is a vertex cover of~$G$.
 Together with the cost correspondence,
 this proves that~$I$ is a \yes-instance
 if and only if~$I'$ is a \yes-instance.

 Finally,
 note that in construction~(b),
 $G'$ contains~$|V|+2$ vertices and~$2|V|+1$ edges.
 Hence,
 the \ETH-based lower bound follows.
\end{proof}

\begin{remark}
 The graph from \cref{constr:planarfvs1}(b) also has neighborhood-diversity two:
 $C_u,C_s$ form one type, and~$V'\setminus \{C_u,C_s\}$ forms the other type.
 In fact,
 adding sufficiently long safe edges between any two non-adjacent vertices turns the graph into clique while preserving correctness,
 which results in neighborhood diversity one
 (but unbounded vertex cover number).
\end{remark}

\begin{remark}
 None of the reductions has unit lengths.
 Yet,
 (a) and (c) can be turned via subdivisions into unit length;
 this is false for~(b) as it would increase the vertex cover number.
\end{remark}

In the remainder,
we show that decreasing any of the bounds $\vcn=2$,
$\fvs=1$,
$\tw=2$,
or $\Delta=3$ by one
leads to graph classes on which \sbnsdAcr{} is polynomial-time solvable.

\begin{proposition}
 \label{prop:P:trees}
 For every instance of \sbnsdAcr{} where the input graph is a tree,
 we can compute a minimum-cost solution in polynomial time.
 Thus,
 \sbnsdAcr{} is polynomial-time solvable on trees.
\end{proposition}
The proof of~\cref{prop:P:trees}
uses the following two reduction~rules.

\begin{rrule}[Irrelevant pair]
 \label{rr:removesatisfied}
 If a terminal pair is connected by a safe path of detour factor at most~$\alpha$,
 delete the terminal pair.
\end{rrule}

\begin{rrule}[Mandatory edge]
 \label{rr:mandatoryedge}
 Let terminal pair~$\{s_r,t_r\}$ be not connected by a safe path of detour factor at most~$\alpha$.
 If an unsafe edge~$e$ lies on all~$s_r$-$t_r$ paths of detour factor at most~$\alpha$,
 then make~$e$ safe.
 If $c(e)$ exceeds the budget,
 then return a trivial \no-instance,
 otherwise decrease the budget by $c(e)$.
\end{rrule}

The correctness of \cref{rr:removesatisfied} is immediate.
The correctness of \cref{rr:mandatoryedge} follows since every feasible solution must upgrade~$e$:
otherwise no safe $s_r$-$t_r$ path of detour factor at most~$\alpha$ can exist.
Thus,
if $\cst(e)>k$, then the instance is a \no-instance;
otherwise, upgrading~$e$ can be fixed in advance and its cost subtracted from the budget.

\begin{proof}[Proof of \cref{prop:P:trees}]
 Let~$I$ be an instance of \sbnsdAcr{} with~$G$ being a tree,
 budget~$k$,
 and let~$\kappa$ be the sum of costs of all unsafe edges.
 Let~$I'$ be the instance obtained from~$I$ by setting the budget to~$\kappa$.
 Apply \cref{rr:removesatisfied,rr:mandatoryedge} exhaustively to~$I'$.
 Since each terminal pair has a unique path,
 each unsafe edge on this path is mandatory and is made safe by \cref{rr:mandatoryedge}.
 Hence,
 every edge made safe by \cref{rr:mandatoryedge} is contained in every feasible solution.
 Consequently,
 once all unsafe edges on the unique path of a terminal pair have been made safe,
 this terminal pair is deleted by \cref{rr:removesatisfied}.
 It follows that no terminal pair is left.
 Let~$F^*$ denote the set of all unsafe edges made safe by \cref{rr:mandatoryedge} during the exhaustive application,
 and let~$\kappa'$ denote the remaining budget in the final instance.
 Note that~$\cst(F^*)=\kappa-\kappa'$.
 Since every edge in~$F^*$ is mandatory and upgrading~$F^*$ satisfies all terminal pairs,
 $\cst(F^*)$ is the minimum cost of any solution.
 Return~\yes{} if~$\cst(F^*)\leq k$,
 and~\no{} otherwise.
\end{proof}

Using~\cref{prop:P:trees},
we also handle graphs of maximum degree at most two:
path components are immediate,
and for each cycle component we either upgrade all unsafe edges or delete one unsafe edge that remains unsafe,
reducing the component to a path.

\begin{observation}
 \label{thm:P:degtwo}
 \sbnsdAcr{} is polynomial-time solvable on graphs of degree at most two.
\end{observation}

\begin{proof}[Proof of \cref{thm:P:degtwo}]
 Let~$I$ be an instance of \sbnsdAcr{} with detour factor~$\alpha$
 where~$G$ has maximum degree at most two.
 Then,
 graph~$G$ decomposes into paths and cycles.
 If a terminal pair is not contained in the same connected component,
 then return \no.
 Otherwise,
 solve each connected component independently.
 If a connected component is a path,
 compute a minimum-cost solution using \cref{prop:P:trees}.
 Now
 consider a connected component~$C$ that is a cycle.
 First,
 consider the solution that upgrades all unsafe edges of~$C$.
 This is feasible for all terminal pairs contained in~$C$.
 We compare this solution with possible solutions where at least one unsafe edge
 is not upgraded.
 Thus,
 for each unsafe edge~$e$ of~$C$,
 consider the case where~$e$ remains unsafe
 and hence no feasible safe path can use~$e$.
 Then,
 every terminal pair~$\{s,t\}$ contained in~$C$
 must be connected along the unique $s$-$t$ path in~$C-\{e\}$.
 This is feasible only if
 $\dist_{C-\{e\}}(s,t)\leq \alpha\cdot \dist_G(s,t)$
 for every terminal pair~$\{s,t\}$ contained in~$C$.
 If this condition holds,
 we delete~$e$ and compute a minimum-cost solution on the path~$C-\{e\}$
 using \cref{prop:P:trees}.
 Among the solution upgrading all unsafe edges
 and all solutions obtained by deleting one unsafe edge,
 we keep one of minimum cost.
 This enumeration is exhaustive since every solution either upgrades all unsafe edges of~$C$
 or leaves at least one unsafe edge~$e$ not upgraded.
 Finally, the optimal solution for~$I$ is composed of the optimal solutions for each of its connected components.
 If the total cost exceeds the budget,
 return~\no,
 otherwise return~\yes.
\end{proof}

\section{Number of Unsafe Edges}

\begin{observation}
 \label{obs:nue:bruteforce}
 \sbnsdAcr{} can be solved in $2^{\nue}\cdot \poly(|I|)$ time and hence is in \FPT{}
 \wpb{}~the number~$\nue$ of unsafe edges.
\end{observation}

\cref{obs:nue:bruteforce} corresponds to testing all possible solutions
in polynomial-time for feasibility.
We show next that,
assuming the \SETH{} to hold,
we cannot do better than this.

\begin{theorem}
 \label{thm:Wh:budget}
 \sbnsdAcr{}, even
 if all unsafe edges induce one component ($\nuc=1$)
 and the input graph has diameter at most four,
 admits
 no $2^{\eps\cdot \nue}\cdot \poly(|I|)$-time algorithm for any $\eps<1$ \unlessSETH{}
 and no problem kernel of size polynomial in~$\nue$ \unlessPK{}.
 Moreover,
 it is \W{2}-hard \wpb{} the budget~$k$.
\end{theorem}

We reduce from the \prob{Hitting Set} problem,
where we are given a set~$X$ of elements,
a set~$\calS$ of nonempty subsets of~$X$,
and an integer~$r\in \N$,
and ask whether there is a subset~$Y\subseteq X$ with~$|Y|\leq r$ such that~$S\cap Y\neq \emptyset$ for every~$S\in\calS$.
\prob{Hitting Set} is \W{2}-hard \wpb{} the solution size~$r$ \cite{DowneyFellows2013} and admits no problem kernel of size polynomial in~$|X|$ \unlessPK{}~\cite{DomLS14}.
Moreover,
\unlessSETH,
\prob{Hitting Set} admits no $2^{\eps\cdot |X|}\cdot \poly(|I|)$-time algorithm for any $\eps<1$~\cite{CyganDLMNOPSW16}.

\begin{construction}
 \label{constr:Wh:budget}
 Let~$I=(X,\calS,r)$ be an instance of \prob{Hitting Set}.
 Construct the graph~$G$ with vertex set~$V\ceq\{v_x\mid x\in X\}\cup\{w_S\mid S\in\calS\}\cup \{v^*\}$,
 safe-edge set~$\Es\ceq\{\{v_x,w_S\}\mid x\in X, S\in\calS: x\in S\}$,
 and unsafe-edge set~$\Eu\ceq\{\{v_x,v^*\}\mid x\in X\}$.
 The length of each edge is one and the cost of each unsafe edge is one.
 Let~$\TPset\ceq \{\{v^*,w_S\}\mid S\in\calS\}$ be the set of terminal pairs.
 Let~$\alpha\ceq 1$
 and~$k\ceq r$.
 \cqed
\end{construction}

Since every unsafe edge is incident with~$v^*$,
the subgraph induced by $\Eu$ is connected
($\nuc=1$).
For every $S\in \calS$,
since $S$ is nonempty,
$w_S$ and~$v^*$ have a common neighbor~$v_x$ with~$x\in S$.
Thus,
every two vertices are connected by a path of length at most four,
and so~$G$ has diameter at most four.

\begin{proof}[Proof of \cref{thm:Wh:budget}]
 Let~$I'=(G=(V,E,\cst,\trt),\TPset,k,\alpha)$ be the instance of \sbnsdAcr{}
 obtained in polynomial time from an instance $I=(X,\calS,r)$ of \prob{Hitting Set}
 via \cref{constr:Wh:budget}.
 For every~$S\in\calS$,
 we have~$\dist_G(w_S,v^*)=2$.
 Moreover,
 $w_S$ is adjacent to~$v_x$ if and only if~$x\in S$.
 Since~$\alpha=1$,
 for every~$F\subseteq \Eu$,
 any safe $w_S$-$v^*$ path of detour factor~$1$ in~$\impro[G]{F}$
 must have length exactly~$2$.
 Hence,
 it contains exactly one inner vertex~$v_x$ with~$x\in S$,
 and the edge~$\{v_x,v^*\}$ must be made safe by~$F$.
 We make the following convention for a set~$Y\subseteq X$ and a set~$F\subseteq \Eu$:
 $x\in Y$ if and only if~$\{v^*,v_x\}\in F$.
 Since every unsafe edge has cost one,
 the convention implies that~$|Y|\leq r$ if and only if~$\cst(F)\leq k$.
 The following equivalences hold.
 Set~$S\in\calS$ is hit by~$Y$ if and only if
 there exists $x\in S\cap Y$.
 By our convention,
 for every~$x\in S$,
 $x\in Y$
 if and only if~$\{v^*,v_x\}\in F$.
 Moreover,
 $\{v^*,v_x\}\in F$ for some~$x\in S$ if and only if
 there is a safe $w_S$-$v^*$ path of detour factor 1
 in~$\impro[G]{F}$.
 It follows that every set~$S\in\calS$ is hit by $Y$
 if and only if every terminal pair admits a safe path of detour factor 1
 in~$\impro[G]{F}$.
 Together with the cost correspondence,
 this proves that~$I$ is a \yes-instance
 if and only if~$I'$ is a \yes-instance.

 Finally,
 note that~$\nue=|X|$.
 Hence,
 both the \SETH-based lower bound and the kernelization lower bound follow.
 \W{2}-hardness transfers since~$k=r$.
\end{proof}

\section{Number of Terminal Pairs}

\begin{theorem}
 \label{thm:wNph}
 \sbnsdAcr{} is weakly \NP-hard even if the input graph is planar,
 of maximum degree at most four,
 and of
 treewidth two,
 all unsafe edges are pairwise disjoint ($\muc=1$),
 and there is only one terminal~pair.
\end{theorem}

We reduce from the weakly \NP-hard \prob{Partition} problem,
where,
given a
multiset $X=\{x_1,\dots,x_N\}$ of at least two positive integers with sum~$\sum_{x_i\in X} x_i = 2T$,
the question is whether there is a partition~$(X_1,X_2)$ of~$X$ such that
$\sum_{x_i\in X_1} x_i = \sum_{x_i\in X_2} x_i = T$.

\begin{construction}
 \label{constr:wNph}
 Let~$I=(X=\{x_1,\dots,x_N\})$ be an instance of \prob{Partition} with sum~$2T$.
 Construct the vertex set~$V=\{v_i\mid i\in\set[0]{N}\}\cup \{a_i,b_i\mid i\in\set{N}\}$ and the edge set as follows:
 For each~$i\in\set{N}$,
 add the safe edges~$\{v_{i-1},a_i\}$ of length~$T$ and~$\{v_{i-1},b_i\}$ of length~$2T$,
 the safe edge~$\{b_i,v_i\}$ of length~$2N x_i$,
 and the unsafe edge~$\{a_i,v_i\}$ of length~$T$ and cost~$x_i$.
 The set~$\TPset$ of terminal pairs consists only of one terminal pair~$\{v_0,v_N\}$.
 Let~$k\ceq T$
 and~$\alpha\ceq 2$.
 \cqed
\end{construction}

The graph is outerplanar:
an outerplanar embedding is obtained by placing~$v_0$ at~$(0,0)$
and, for each~$i\in\set{N}$,
placing~$v_i$ at~$(2i,0)$,
$a_i$ at~$(2i-1,1)$,
and~$b_i$ at~$(2i-1,-1)$.
Since the graph is outerplanar and contains a cycle,
it has treewidth two.
It has maximum degree four:
as~$N\geq 2$,
each vertex~$v_i$ with~$i\in\set{N-1}$ has degree four,
whereas each vertex~$a_i$ and~$b_i$ with~$i\in\set{N}$,
as well as~$v_0$ and~$v_N$,
has degree two.
Finally,
the set~$\{\{a_i,v_i\}\mid i\in\set{N}\}$ of unsafe edges
consists of pairwise vertex-disjoint edges,
and hence~$\muc=1$.

\begin{proof}[Proof of \cref{thm:wNph}]
 Let~$I'=(G=(V,E,\cst,\trt),\TPset,k,\alpha)$ be the instance of \sbnsdAcr{}
 obtained in polynomial time from an instance $I=(X=\{x_1,\dots,x_N\})$ of \prob{Partition} with sum~$2T$
 via \cref{constr:wNph}.
 The shortest $v_0$-$v_N$ path is of length~$2N T$
 by using the unsafe branch when traversing from~$v_{i-1}$ to~$v_i$ via $a_i$ for each~$i\in\set{N}$.
 Hence,
 since~$\alpha=2$,
 every feasible path in a solution must be of length at most~$4NT$.
 We show that~$I$ is a \yes-instance if and only if~$I'$ is a \yes-instance.

 For the forward direction,
 let~$(X_1,X_2)$ be a solution to~$I$.
 We claim that~$F\ceq \{\{a_i,v_i\}\mid x_i\in X_1\}$ is a solution to~$I'$.
 We have that~$\cst(F)=\sum_{x_i\in X_1} \cst(\{a_i,v_i\}) = \sum_{x_i\in X_1} x_i = T\leq k$.
 Moreover,
 consider the shortest safe $v_0$-$v_N$ path in~$\impro[G]{F}$.
 For each~$x_i\in X_1$,
 it uses the branch traversing from~$v_{i-1}$ to~$v_i$ via $a_i$,
 since the unsafe edge~$\{a_i,v_i\}$ is upgraded by~$F$;
 this branch has length~$2T$.
 For each~$x_i\in X_2$,
 it uses the branch traversing from~$v_{i-1}$ to~$v_i$ via $b_i$,
 which has length~$2T+2N x_i$.
 Altogether,
 the path has length~$2NT + 2N\sum_{x_i\in X_2} x_i  = 4NT$.
 Thus,
 $F$ is a solution to~$I'$.

 For the backward direction,
 let~$F$ be a solution to~$I'$.
 We claim that~$(X_1,X_2)$,
 where
 $X_1\ceq\{x_i\mid \{a_i,v_i\}\in F\}$ and~$X_2\ceq X\setminus X_1$,
 is a solution to~$I$.
 Since $F$ is a solution,
 the shortest safe $v_0$-$v_N$ path in~$\impro[G]{F}$ is of length~$2NT + 2N\sum_{x_i\in X_2} x_i\leq 4NT$.
 It follows that $\sum_{x_i\in X_2} x_i\leq T$.
 Moreover,
 we know that~$\sum_{x_i\in X_1} x_i = \sum_{\{a_i,v_i\}\in F} \cst(\{a_i,v_i\}) = \cst(F)\leq k= T$.
 Finally,
 since~$\sum_{x_i\in X_2} x_i + \sum_{x_i\in X_1} x_i = 2T$,
 $(X_1,X_2)$ is a solution to~$I$.
\end{proof}

\begin{remark}
  Note that $\fvs$,
  and hence $\fes$,
  is unbounded in the constructed graph.
  Yet,
  the graph has cutwidth two,
  witnessed by an ordering such as $(\dots,v_{i-1},a_i,b_i,v_i,\dots)$,
  with exactly two edges between any two consecutive vertices.
\end{remark}

Note that in the previous reduction there is only \emph{one} terminal pair.
For unary encodings,
the reduction no longer applies;
however,
we obtain \W{1}-hardness in this case.

\begin{theorem}
 \label{thm:Wh:tps}
 \sbnsdAcr{} is \W{1}-hard \wpb{} the number~$\ntp$ of terminals combined with the budget~$k$, even if $\nuc=1$,
 all edges are unsafe,
 and all shortest terminal paths are of length three.
\end{theorem}

We reduced from the \prob{Multicolored Clique} problem,
which asks,
given a graph~$G=(V,E)$ where~$V=V_1\uplus\cdots\uplus V_r$ is partitioned into $r$ color classes
and at least one vertex from~$V_i$ is adjacent to a vertex from~$V_j$ for every distinct~$i,j\in\set{r}$,
whether there is a vertex set~$W$ such that~$W$ forms a clique in~$G$
and~$|W\cap V_i|=1$ for every~$i\in\set{r}$.
\prob{Multicolored Clique} is \W{1}-hard \wpb{}~$r$ \cite{FellowsHRV2009}.

\begin{construction}
 \label{constr:Wh:tps}
 Let~$I=(G=(V=V_1\uplus\cdots\uplus V_r,E))$ be an instance of \prob{Multicolored Clique}.
 Let~$G'$ be initially a copy of~$G$.
 Make each edge in~$G'$ unsafe with cost 1 and length 1.
 For each color class add a verifier vertex~$w_i$ and make it adjacent by unsafe edges of cost $r^3$ and length 1 with all vertices from~$V_i$.
 Let~$\TPset\ceq \{\{w_i,w_j\}\mid 1\leq i<j\leq r\}$ be the set of terminal pairs.
 Let $\alpha\ceq 1$ and~$k\ceq r\cdot r^3 + \binom{r}{2}$.
 \cqed
\end{construction}

\begin{proof}[Proof of \cref{thm:Wh:tps}]
 Let~$I'=(G'=(V',E',\cst,\trt),\TPset,k,\alpha)$ be the instance of \sbnsdAcr{}
 obtained in polynomial time from an instance $I=(G=(V=V_1\uplus\cdots\uplus V_r,E))$ of \prob{Multicolored Clique}
 via \cref{constr:Wh:tps}.
 We show that~$I$ is a \yes-instance if and only if $I'$ is a \yes-instance.

 For the forward direction,
 let~$W=\{v_1^*,\dots,v_r^*\}$ be solution to~$I$ with~$v_i^*\in V_i$ for each~$i\in\set{r}$.
 We claim that~$F\ceq \{\{w_i,v_i^*\}\mid i\in\set{r}\}\cup\{\{v_i^*,v_j^*\}\mid i,j\in\set{r},\, i< j\}$ is a solution to~$I'$.
 We have that~$\cst(F)=\sum_{i\in\set{r}} \cst(\{w_i,v_i^*\}) + \sum_{i,j\in\set{r},\, i< j} \cst(\{v_i^*,v_j^*\}) = r\cdot r^3 + \binom{r}{2} = k$.
 Moreover,
 for each~$i,j\in\set{r}$,
 $i< j$,
 there is a safe~$w_i$-$w_j$ path in~$\impro[G']{F}$ of length three traversing from $w_i$ to~$v_i^*$ to $v_j^*$ to~$w_j$ of length~$\dist_{G'}(w_i,w_j)$.
 Hence,
 $F$ is a solution to~$I'$.

 For the backward direction,
 let~$F$ be a solution to~$I'$.
 By construction,
 for each color class~$i\in\set{r}$,
 $V_i$ separates~$w_i$ in~$G'$ from the rest of the graph.
 Since~$w_i$ appears in at least one terminal pair,
 at least one edge incident with~$w_i$
 must be upgraded by~$F$ at cost~$r^3$.
 As there are~$r$ color classes,
 at least~$r\cdot r^3$ budget is spent on these upgrades.
 Thus,
 the remaining budget is at most~$\binom{r}{2}<r^3$,
 and hence no further edge incident with any~$w_i$ can be upgraded.
 Hence,
 for every~$i\in\set{r}$,
 there is exactly one upgraded edge incident with~$w_i$ in~$F$.
 Let this edge be~$\{w_i,v_i^*\}$ with~$v_i^*\in V_i$.
 We claim that~$W\ceq \{v_1^*,\dots,v_r^*\}$ is a solution to~$I$.
 Consider distinct~$i,j\in\set{r}$.
 By the assumption on the input instance,
 at least one vertex from~$V_i$ neighbors a vertex from~$V_j$,
 and hence~$\dist_{G'}(w_i,w_j)=3$.
 Since~$\alpha=1$,
 every feasible safe $w_i$-$w_j$ path in $\impro[G']{F}$
 must have exactly three edges.
 Since~$\{w_i,v_i^*\}$ is the only upgraded edge incident with~$w_i$
 and~$\{w_j,v_j^*\}$ is the only upgraded edge incident with~$w_j$,
 the middle edge of such a path must be~$\{v_i^*,v_j^*\}$.
 Thus,
 $\{v_i^*,v_j^*\}\in E$.
 Since this holds for every pair of distinct color classes,
 it follows that~$W$ induces a clique.
 Hence,
 $W$ is a solution to~$I$.

 Since
 $|\TPset|=\binom{r}{2}$ and~$k=r\cdot r^3 + \binom{r}{2}$ depend only on~$r$,
 \W{1}-hardness for the combined parameter~$\ntp+k$ follows.
\end{proof}

\section{Feedback Edge Number}

\begin{theorem}
 \label{thm:FESTP:linker}
 Any instance of \sbnsdAcr{} can be mapped in polynomial time to an equivalent instance of \sbnsdAcr{} with~$O(\fes + \ntp)$
 vertices and edges.  %
\end{theorem}

Note that together with \cref{obs:nue:bruteforce},
it follows that \sbnsdAcr{} is \FPT{} \wpb{} $\fes + \ntp$.
\cref{thm:FESTP:linker} consists of carefully deleting leaves and compressing long paths
in the tree~$G-X$,
where~$X$ denotes a minimum feedback edge set.
Such an approach has already proved useful
(see, e.g., \cite{FeldmannLampis2025,KellerhalsKoana2022}).
We call a vertex in~$G=(V,E)$ with terminal pairs~$\TPset$ and a feedback edge set~$X\subseteq E$ \emph{important}
if it is contained in a terminal pair,
incident with an edge of~$X$,
or is of degree at least three in~$G-X$,
and \emph{unimportant} otherwise.

\begin{rrule}[Leaf]
 \label{rr:leafpruning}
 If a leaf in~$G-X$ is unimportant,
 delete it.
\end{rrule}

\begin{proof}[Correctness proof of \cref{rr:leafpruning}]
  Let~$v$ be an unimportant leaf of~$G-X$.
  Then,
  $v$ is neither contained in a terminal pair nor incident with an edge of~$X$.
  Thus,
  $v$ is a non-terminal leaf of~$G$
  and hence not contained in any path between terminals.
  Deleting~$v$ preserves all terminal-pair distances,
  all feasible safe paths,
  and the set of feasible solutions.
\end{proof}

For terminal pairs~$\TPset$ and a feedback edge set~$X\subseteq E$,
we call a path~$P$ in~$G-X$ a \emph{corridor} if it contains at least three vertices,
each of its endpoints is important,
and all its inner vertices are unimportant and of degree two in~$G-X$.
We exploit that any $s$-$t$ path connecting a terminal pair~$\{s,t\}$ contains either all edges of such a corridor or none.
Hence,
we can replace such a corridor by a path with at most two edges
thereby preserving all relevant information
(lengths and costs)
for the instance.
However,
we need to carefully distinguish the type of the corridor in terms of contained safe and unsafe edges and whether its endpoints are connected by a feedback edge.
Note that since~$P$ is a subgraph of~$G-X$,
if the endpoints of~$P$ are adjacent in~$G$,
then the edge belongs to~$X$.
Let~$v,w$ be the endpoints of corridor~$P$.
We call~$P$ \emph{unsafe} if it contains at least one unsafe edge,
and \emph{safe} otherwise.
We prove next that we can replace corridors with many edges by paths with at most two edges.

\begin{lemma}
 \label{lem:corrrepl}
 Let~$I$ be an instance of \sbnsdAcr{} with graph~$G=(V,E)$ and with feedback edge set~$X\subseteq E$.
 Let~$P$ be a corridor with endpoints~$v,w$.
 Let~$I'$ be the instance with graph~$G'=(V',E')$ obtained from~$I$ only by replacing~$P$ with a path~$P'$ with endpoints~$v,w$ such that
 $\len(P')=\len(P)$,
 and
 $P'$ consists of
 either two edges and one unimportant inner vertex if~$\{v,w\}\in E$,
 or of one edge otherwise.
 Moreover,
 $P$ is safe if and only if~$P'$ is safe,
 and if both are unsafe,
 $\sum_{e\in \EG(P)\cap \Eu} \cst(e) = \sum_{e\in \EG(P')\cap \Eu'} \cst(e)$,
 where~$\Eu'\subseteq E'$ is the set of unsafe edges in~$G'$.
 Then~$I$ is a \yes-instance if and only if~$I'$ is a \yes-instance.
\end{lemma}

\begin{proof}
 Since every inner vertex of~$P$ and of~$P'$ is unimportant,
 none of them is in a terminal pair.
 Since the inner vertices of~$P$ and~$P'$ have degree two in~$G$ and~$G'$,
 respectively,
 every path connecting a terminal pair contains either all or none of the edges of~$P$ and~$P'$.
 Since~$\len(P)=\len(P')$,
 it follows that~$\dist_G(s,t) = \dist_{G'}(s,t)$ for every terminal pair~$\{s,t\}\in\TPset$.
 Finally,
 note that~$G'$ is a simple graph and hence a valid input to \sbnsdAcr{}
 since~$P'$ has two edges in the case of~$\{v,w\}\in E$.

 For the forward direction,
 let~$F\subseteq E$ be a cost-minimal solution to~$I$.
 Let~$F_P \ceq F\cap E(P)$.
 We know that either~$F_P=\emptyset$ or~$F_P=E(P)\cap \Eu$.
 Let~$F_{P'}\ceq E(P')\cap \Eu'$.
 Let~$F'$ be~$F$ if~$F_P=\emptyset$,
 and~$(F\setminus F_P)\cup F_{P'}$ otherwise.
 By the cost equality for~$P$ and~$P'$,
 we have~$\cst(F')=\cst(F)\leq k$.
 We claim that~$F'$ is a solution to~$I'$.
 Let~$\{s,t\}\in\TPset$ be a terminal pair.
 Let~$Q$ be a safe $s$-$t$ path in~$\impro[G]{F}$ of detour factor at most~$\alpha$.
 If~$Q$ contains no edge from~$P$,
 then~$Q$ is also a safe $s$-$t$ path in~$\impro[G']{F'}$,
 since $\dist_G(s,t) = \dist_{G'}(s,t)$.
 If~$Q$ contains an edge from~$P$,
 then~$P$ is a subpath of~$Q$.
 Let~$Q'$ be the path in~$G'$ obtained from~$Q$ by replacing~$P$ by~$P'$.
 Note that~$\len(Q)=\len(Q')$.
 Moreover,
 $Q'$ is safe since if all unsafe edges of~$P$ are upgraded by~$F$,
 so are all unsafe edges of~$P'$ upgraded by~$F'$.
 With $\dist_G(s,t) = \dist_{G'}(s,t)$,
 it follows that~$Q'$ is safe and of detour factor at most~$\alpha$.
 Since this holds for every terminal pair,
 $F'$ is a solution to~$I'$.

 For the backward direction,
 let~$F'\subseteq E'$ be a cost-minimal solution to~$I'$.
 Let~$F_{P'} \ceq F'\cap E(P')$.
 We know that either~$F_{P'}=\emptyset$ or~$F_{P'}=E(P')\cap \Eu'$.
 Let~$F_{P}\ceq E(P)\cap \Eu$.
 Let~$F$ be~$F'$ if~$F_{P'}=\emptyset$,
 and~$(F'\setminus F_{P'})\cup F_{P}$ otherwise.
 By the cost equality for~$P$ and~$P'$,
 we have~$\cst(F)=\cst(F')\leq k$.
 We claim that~$F$ is a solution to~$I$.
 Let~$\{s,t\}\in\TPset$ be a terminal pair.
 Let~$Q'$ be a safe $s$-$t$ path in~$\impro[G']{F'}$ of detour factor at most~$\alpha$.
 If~$Q'$ contains no edge from~$P'$,
 then~$Q'$ is also a safe $s$-$t$ path in~$\impro[G]{F}$,
 since $\dist_G(s,t) = \dist_{G'}(s,t)$.
 If~$Q'$ contains an edge from~$P'$,
 then~$P'$ is a subpath of~$Q'$.
 Let~$Q$ be the path in~$G$ obtained from~$Q'$ by replacing~$P'$ by~$P$.
 Note that~$\len(Q')=\len(Q)$.
 Moreover,
 $Q$ is safe since if all unsafe edges of~$P'$ are upgraded by~$F'$,
 so are all unsafe edges of~$P$ upgraded by~$F$.
 With $\dist_G(s,t) = \dist_{G'}(s,t)$,
 it follows that~$Q$ is safe and of detour factor at most~$\alpha$.
 Since this holds for every terminal pair,
 $F$ is a solution to~$I$.
\end{proof}

If the endpoints of a corridor are connected by a feedback edge,
we compare the corridor with this alternative connection.
To this end,
we introduce the following notation.
We call a path~$R$ with endpoints~$v,w$ and no important inner vertex \emph{irrelevant}
if there is another path~$R'$ with endpoints~$v,w$ in
$G-(V(R)\setminus\{v,w\})$
and no important inner vertex
such that
if~$R$ is safe,
then~$R'$ is safe and $\len(R')\leq \len(R)$;
if~$R$ is unsafe,
then~$\len(R')\leq \len(R)$
and $\sum_{e\in \EG(R')\cap E_u} c(e) \leq \sum_{e\in \EG(R)\cap E_u} c(e)$.

\begin{lemma}
 \label{lem:irrelevant}
 Let~$I$ be an instance of \sbnsdAcr{} with an irrelevant path~$R$.
 Let~$I'$ be the instance obtained from~$I$ by deleting all edges and inner vertices of~$R$.
 Then,
 $I$ is a \yes-instance if and only if~$I'$ is a \yes-instance.
\end{lemma}

\begin{proof}
 Let~$v,w$ denote the endpoints of~$R$ and let~$R'$ be a witness such that~$R$ is irrelevant.
 Since no inner vertex of~$R$ is a terminal and every inner vertex has degree two,
 every terminal-to-terminal path contains either all or none of the edges of~$R$.
 Let~$G'$ denote the graph in~$I'$.
 Since~$G'$ is a subgraph of~$G$,
 we have~$\dist_G(s,t)\leq \dist_{G'}(s,t)$.
 Conversely,
 every path using~$R$ can replace it by~$R'$ without increasing its length,
 as~$\len(R')\leq \len(R)$.
 Thus,
 $\dist_G(s,t)=\dist_{G'}(s,t)$ for every terminal pair~$\{s,t\}\in\TPset$.
 Next we show that $I$ is a \yes-instance if and only if~$I'$ is a \yes-instance.
 Regarding the backward direction:
 Since~$G'$ is a subgraph of~$G$ and distances are preserved,
 any solution for $I'$ is also feasible for $I$.

 For the forward direction,
 let~$F$ be a solution to~$I$.
 If~$R$ is safe,
 then every feasible safe path using~$R$ can replace~$R$ by the safe path~$R'$,
 which is no longer since $\len(R')\leq \len(R)$.
 All paths not using~$R$ remain unchanged.
 Thus,
 $F$ is also a solution to~$I'$.
 Now assume that~$R$ is unsafe.
 If~$F$ does not contain all unsafe edges of~$R$,
 then~$F'\ceq F\setminus (F\cap \EG(R))$ is also a solution to~$I'$
 since no safe path in~$\impro[G]{F}$ uses an edge from~$R$.
 Otherwise,
 let $F'$ be obtained from~$F$ by replacing all unsafe edges of~$R$ by all unsafe edges of~$R'$.
 Since $\sum_{e\in \EG(R')\cap E_u} c(e) \leq \sum_{e\in \EG(R)\cap E_u} c(e)$,
 the cost of~$F'$ is at most the cost of~$F$.
 Since~$F$ is a solution to~$I$,
 for every terminal pair~$\{s,t\}\in \TPset$,
 there is a safe path~$Q$ in~$\impro[G]{F}$ of detour factor at most~$\alpha$.
 If~$Q$ does not contain~$R$,
 then it is also a safe path in~$\impro[G']{F'}$.
 If~$Q$ contains~$R$,
 then the $s$-$t$ path~$Q'$ obtained from~$Q$ by replacing~$R$ by~$R'$
 is a safe path in~$\impro[G']{F'}$ of length at most~$\len(Q)$,
 since~$\len(R')\leq\len(R)$.
 Hence,
 $F'$ is also a solution to~$I'$.
\end{proof}

We proceed towards our reduction rules.
We call a safe corridor~$P$ \emph{contractible}
if
(i)~$\{v,w\}\not\in E$,
(ii)~$e=\{v,w\}\in X\cap \Es$, or
(iii)~$e=\{v,w\}\in X\cap \Eu$ and~$\len(P)\leq \trt(e)$,
and \emph{compressible} if it contains at least four vertices,
$e=\{v,w\}\in X\cap \Eu$,
and $\len(P)> \trt(e)$.
A safe corridor can be neither contractible nor compressible,
that is,
when it has three vertices,
$e=\{v,w\}\in X\cap \Eu$,
and $\len(P)> \trt(e)$.

\begin{rrule}[Safe Corridor Contraction]
 \label{rr:safepathcontraction}
 Let~$P$ be a contractible safe corridor with endpoints~$v,w$.
 \begin{enumerate}
  \item If~$\{v,w\}\not\in E$, then remove all edges and inner vertices from~$P$ and add the edge~$e=\{v,w\}$ to~$\Es$ and set~$\trt(e)=\len(P)$.
  \item If~$e=\{v,w\}\in X\cap \Es$, then remove~$e$ and all edges and inner vertices from~$P$,
  add a new safe edge~$e'$ and set~$\trt(e')=\min\{\trt(e),\len(P)\}$.
  \item If $e=\{v,w\}\in X\cap \Eu$ and~$\len(P)\leq \trt(e)$,
  then remove~$e$ and all edges and inner vertices from~$P$ and add the edge~$e'=\{v,w\}$ to~$\Es$ and set~$\trt(e')=\len(P)$.
 \end{enumerate}
\end{rrule}

\cref{rr:safepathcontraction} is correct since an unsafe alternative not shorter than the safe corridor is irrelevant,
and among two safe options,
the longer is irrelevant.

\begin{proof}[Correctness proof of \cref{rr:safepathcontraction}]
 The correctness of the first case follows directly from~\cref{lem:corrrepl}.
 For the second case,
 if~$\trt(e)\leq \len(P)$,
 then $P$ is irrelevant and correctness follows from~\cref{lem:irrelevant}.
 Otherwise,
 the single-edge path~$e$ is irrelevant.
 Hence,
 the correctness follows from~\cref{lem:irrelevant} together with~\cref{lem:corrrepl}.
 For the third case,
 since~$\len(P)\leq \trt(e)$,
 the single-edge path~$e$ is irrelevant.
 Hence,
 the correctness follows from~\cref{lem:irrelevant} together with~\cref{lem:corrrepl}.
\end{proof}

The next rule treats a shorter unsafe alternative to the safe corridor.
In this case,
we have to keep both,
but can compress the safe corridor into a path with two edges that preserves the length of the~corridor.
The correctness follows directly from~\cref{lem:corrrepl}.

\begin{rrule}[Safe Corridor Compression]
 \label{rr:safepathcompression}
 Let~$P$ be a compressible safe corridor with endpoints~$v,w$ and let $e=\{v,w\}\in X\cap \Eu$.
 Then remove all edges and inner vertices from~$P$,
 add a new vertex~$x_{v,w}$ and the edges~$e'=\{v,x_{v,w}\},e''=\{x_{v,w},w\}$ to~$\Es$,
 and set~$\trt(e')=\len(P)-1$ and~$\trt(e'')=1$.
\end{rrule}

Let~$P$ be an unsafe corridor with total cost $M\ceq \sum_{e\in \EG(P)\cap \Eu} \cst(e)$.
We call~$P$ \emph{contractible}
if
(i)~$\{v,w\}\not\in E$,
(ii)~$e=\{v,w\}\in X\cap \Es$ and $\len(P)\geq \trt(e)$, or
(iii)~$e=\{v,w\}\in X\cap \Eu$~and
\begin{align}
  (\len(P)\leq \trt(e) \land M\leq \cst(e)) \lor (\len(P)\geq \trt(e) \land M\geq \cst(e)).
  \label{eq:unsafebreaker}
\end{align}
We call~$P$ \emph{compressible}
if it contains at least four vertices
and
either
(i)~$e=\{v,w\}\in X\cap \Es$ and~$\len(P)<\trt(e)$, or
(ii)~$e=\{v,w\}\in X\cap \Eu$ and \eqref{eq:unsafebreaker} does not hold.
An unsafe corridor can be neither contractible nor compressible,
that is,
when it has three vertices and the conditions (i) or (ii) for compressible corridors hold.

\begin{rrule}[Unsafe Corridor Contraction]
 \label{rr:unsafepathcontraction}
 Let~$P$ be a contractible unsafe corridor with endpoints~$v,w$.
 Let~$M\ceq \sum_{e\in \EG(P)\cap \Eu} \cst(e)$ be the total cost of~$P$.
 \begin{enumerate}
  \item If~$\{v,w\}\not\in E$, then remove all edges and inner vertices from~$P$ and add the edge~$e=\{v,w\}$ to~$\Eu$ and set~$\trt(e)=\len(P)$ and~$\cst(e)=M$.
  \item If $e=\{v,w\}\in X\cap \Es$ and $\len(P)\geq \trt(e)$,
  then remove~$e$ and all edges and inner vertices from~$P$,
  add the safe edge~$e'=\{v,w\}$ and set~$\trt(e')=\trt(e)$.
  \item If~$e=\{v,w\}\in X\cap \Eu$ and \eqref{eq:unsafebreaker} holds,
  then remove~$e$ and all edges and inner vertices from~$P$,
  add a new unsafe edge $e'$ and let~$\cst(e')\ceq \min\{M,\cst(e)\}$ and~$\trt(e')\ceq\min\{\trt(e),\len(P)\}$.
 \end{enumerate}
\end{rrule}

\cref{rr:unsafepathcontraction} is correct since when one unsafe connection weakly dominates the other
(neither longer nor more costly),
the dominated one is irrelevant for both shortest paths and upgrade costs.

\begin{proof}[Correctness proof of \cref{rr:unsafepathcontraction}]
 The correctness of the first case follows directly from~\cref{lem:corrrepl}.
 For the second case,
 since~$\len(P)\geq \trt(e)$,
 $P$ is irrelevant and correctness follows from~\cref{lem:irrelevant}.
 Now consider the third case.
 If $\len(P)\leq \trt(e)$ and $M\leq \cst(e)$,
 then
 the single-edge path~$e$ is irrelevant.
 Hence,
 the correctness follows from~\cref{lem:irrelevant} together with~\cref{lem:corrrepl}.
 If $\len(P)\geq \trt(e)$ and $M\geq \cst(e)$,
 then~$P$ is irrelevant and correctness follows from~\cref{lem:irrelevant}.
 These two cases are exactly the two alternatives in~\eqref{eq:unsafebreaker}.
\end{proof}

The next rule treats two incomparable unsafe paths with the same endpoints.
In this case,
we have to keep both,
but can compress the unsafe corridor into a path with two edges
that preserves the length of the corridor and the cost to upgrade all of its unsafe edges.
The correctness follows directly from~\cref{lem:corrrepl}.

\begin{rrule}[Unsafe Corridor Compression]
 \label{rr:unsafepathcompression}
 Let~$P$ be a compressible unsafe corridor with endpoints~$v,w$.
 Let~$M\ceq \sum_{e\in \EG(P)\cap \Eu} \cst(e)$ be the total cost of~$P$.
 Then remove all edges and inner vertices from~$P$,
 add a new vertex~$x_{v,w}$ and the edges~$e'=\{v,x_{v,w}\}$ to~$\Eu$ and $e''=\{x_{v,w},w\}$ to~$\Es$.
 Set~$\trt(e')=\len(P)-1$,
 $\cst(e')=M$,
 and~$\trt(e'')=1$.
\end{rrule}

For easier reference,
we summarize the four preceding reduction rules in the following reduction rule.
Its correctness follows from the correctness of the four individual rules.

\begin{rrule}[Corridor]
 \label{rr:corridor}
 Let~$P$ be a corridor.
 If~$P$ is safe,
 then apply \cref{rr:safepathcontraction} if it is contractible
 and \cref{rr:safepathcompression} if it is compressible.
 If~$P$ is unsafe,
 then apply \cref{rr:unsafepathcontraction} if it is contractible
 and \cref{rr:unsafepathcompression} if it is compressible.
\end{rrule}

Each corridor replacement preserves the relevant $v$-$w$ connections.
Safe and unsafe corridors represent length-$\len(P)$ connections
that are already safe and that can be made safe at upgrade cost~$M$,
respectively.
If the endpoints are joined by a feedback edge~$e$,
then $e$ contributes a length-$\trt(e)$ connection that is either already safe or
that can be made safe at upgrade cost~$\cst(e)$.
The contraction rules delete dominated connections,
while the compression rules keep two incomparable connections by representing
the corridor as a two-edge path.
Since every $s$-$t$ path for terminal pair~$\{s,t\}$
uses either the whole corridor or none of it,
these replacements preserve feasibility and optimum cost.

\begin{proof}[Proof of~\cref{thm:FESTP:linker}]
 Let~$X$ be a minimum feedback edge set of~$G$ with $|X|=\fes$
 and let~$T$ denote the forest obtained from~$G-X$ when \cref{rr:leafpruning} is applied exhaustively.
 We mark all important vertices and keep these marks fixed throughout the reduction.
 Now,
 apply
 \cref{rr:corridor} exhaustively.
 Whenever a rule deletes an edge of~$X$,
 we remove it from the current feedback edge set.
 New edges replacing subpaths of $T$ are treated as forest edges.
 Let~$T'$ denote the forest obtained from~$T$ by the above replacements,
 $X'\subseteq X$ the remaining feedback edges
 (recall that applying corridor contractions can delete feedback edges as well),
 and let~$G'$ denote the graph~$T'$ together with the edges from~$X'$.
 Let~$\VG(T')=V_1\cup V_2^+\cup V_2^- \cup V_{\geq 3}$,
 where we denote by $V_1$ the set of leaves,
 by $V_2^+$ the set of important degree-2 vertices,
 by $V_2^-$ the set of unimportant degree-2 vertices,
 and by~$V_{\geq 3}$ the set of vertices of degree at least three.
 We know that~$|V_1|+|V_2^+|\leq 2\fes+2\ntp$,
 and it is well known that~$|V_{\geq 3}|\leq |V_1|$.
 Next we show that~$|V_2^-|\leq \fes$.
 After the execution,
 every remaining unimportant degree-two vertex is the unique inner vertex of a corridor with three vertices whose endpoints are important and adjacent by a remaining feedback edge.
 Since~$T'$ is a forest,
 for each feedback edge~$\{v,w\}\in X'$ there is at most one $v$-$w$ path in~$T'$.
 Hence, this defines an injective mapping from unimportant degree-two vertices to feedback edges in~$X'$,
 and thus~$|V_2^-|\leq |X'|\leq \fes$.
 Hence,
 we get that~$|\VG(G')|=|\VG(T')|=|V_1|+|V_2^+|+|V_2^-|+|V_{\geq 3}|\leq |V_1|+|V_2^+|+|V_2^-|+|V_1| \leq 5\fes+4\ntp$.
 Finally,
 $|\EG(G')|= |\EG(T')| + |X'| \leq 6\fes+4\ntp$.
\end{proof}

\section{ILP, Preprocessing, and Cut Generation}

For every~$r\in\set{\ntp}$,
we denote by~$A_r^\alpha$ the set of all \emph{$\alpha$-admissible} arcs,
where,
for an edge~$e=\{v,w\}$ and terminal pair~$\{s_r,t_r\}$,
arc~$(v,w)$ is $\alpha$-admissible for terminal pair~$\{s_r,t_r\}$ if
$\dist(s_r,v) + \trt(e) + \dist(w,t_r) \leq \alpha\cdot\dist(s_r,t_r)$.
For a vertex~$v\in V$,
we denote by~$\outarcs_r^\alpha(v)$ and~$\inarcs_r^\alpha(v)$ the sets of outgoing and incoming arcs at~$v$ in~$A_r^\alpha$,
respectively.
We say that an edge is~$\alpha$-admissible for terminal pair~$\{s_r,t_r\}$
if at least one of its two associated arcs is $\alpha$-admissible for terminal pair~$\{s_r,t_r\}$.
Let~$E_r^\alpha$ denote the set of all $\alpha$-admissible edges for terminal pair~$\{s_r,t_r\}$.
Finally,
let~$G_r^\alpha\ceq G[E_r^\alpha]$ denote the $\alpha$-admissible graph for terminal pair~$\{s_r,t_r\}$.
We use the following ILP as our baseline exact formulation.
\begin{subequations}
\label{eq:ilp-base}
\begin{align}
\min \quad
& \sum_{e\in \Eu} \cst(e)\cdot  x_e
\label{eq:ilp-obj}
\\
\text{s.t.}\quad
& \sum_{a\in \outarcs_r^\alpha(v)} f_a^r
-
\sum_{a\in \inarcs_r^\alpha(v)} f_a^r
=
\begin{cases}
1, & \cif{}v=s_r,\\
-1, & \cif{}v=t_r,\\
0, & \cotw,
\end{cases}
&& \forall r\in\set{\ntp},\ \forall v\in V,
\label{eq:ilp-flow}
\\
&
\sum_{a\in A_r^\alpha\cap \{(v,w),(w,v)\}} f_{a}^r
\le
x_e,
&& \forall r\in\set{\ntp},\ \forall \{v,w\}\in \Eu,
\label{eq:ilp-activation}
\\
& \sum_{a\in A_r^\alpha} \trt(a)\cdot f_a^r
\le
\alpha\cdot \dist_G(s_r,t_r),
&& \forall r\in\set{\ntp},
\label{eq:ilp-length}
\\
& 0 \le f_a^r \le 1,
&& \forall r\in\set{\ntp},\ \forall a\in A_r^\alpha.
\label{eq:ilp-f-domain}
\\
& x_e \in \{0,1\},
&& \forall e\in \Eu,
\label{eq:ilp-x-domain}%
\end{align}%
\end{subequations}%
Constraints~\eqref{eq:ilp-flow}--\eqref{eq:ilp-f-domain}
ensure that for each terminal pair~$\{s_r,t_r\}$,
a safe $s_r$-$t_r$ path exists,
that can use upgraded unsafe edges \eqref{eq:ilp-activation},
but must have detour factor at most~$\alpha$ \eqref{eq:ilp-length}.

\subsection{FES-based Preprocessing}
\label{sec:fesbasedpreprocessing}

We apply the reduction rules that we described before
together with the following two,
where the second is a direct combination of \cref{rr:removesatisfied,rr:mandatoryedge}:

\begin{rrule}[Irrelevant edge]
 \label{rr:alpha}
 If an edge is not $\alpha$-admissible for any terminal pair,
 delete it.
\end{rrule}

\begin{rrule}[Unique path]
 \label{rr:uniquepath}
 If a terminal pair is not connected by a safe path of detour factor at most~$\alpha$
 but connected by a unique path of detour factor at most~$\alpha$ containing at least one unsafe edge,
 then let~$C$ denote the sum of costs of all unsafe edges on this path.
 If~$C$ exceeds the budget, return a trivial \no-instance,
 otherwise make all unsafe edges on the path safe,
 decrease the budget by~$C$,
 and delete the terminal pair.
\end{rrule}

We apply each reduction rule exhaustively before moving to the next one.
Since one reduction rule can make another rule applicable,
their ordering matters.
\cref{fig:shrink} shows how we arrange our reduction rules.
\begin{figure}[t]
 \centering
 \begin{tikzpicture}
  \def\xr{1.75}
  \def\yr{1}
  \tikzpreamble{}

  \newcommandx{\rnode}[5][4=AaZz]{%
    \ifstrequal{#4}{AaZz}{%
      \node (#1) at (#2*\xr,#3*\yr)[anchor=west,label={[label distance=-4pt]-90:{\tiny(\nameref{#5})}}]{\Cref{#5}};
    }{%
      \node (#1) at (#2*\xr,#3*\yr)
        [anchor=west,
        label={[label distance=-1em]-90:{\tiny\shortstack{(\nameref{#4}\\\nameref{#5})}}}]
        {\Cref{#4} \& \Cref{#5}};

    }
  }
  \clip (-1*\xr,-0.85*\yr) rectangle (6.9*\xr,0.95*\yr);

  \rnode{a}{0}{0}{rr:removesatisfied};
  \rnode{b}{1}{0}{rr:leafpruning};
  \rnode{c}{2}{0}{rr:uniquepath};
  \rnode{d}{3}{0}{rr:leafpruning};
  \rnode{e}{4.125}{0}{rr:corridor};
  \rnode{f}{6}{0.4}{rr:mandatoryedge};
  \rnode{g}{6}{-0.4}{rr:alpha};

  \node (s) at (-1*\xr,-0.5*\yr)[anchor=west,draw]{\texttt{start}};

  \foreach \x/\y in {s/a.west,a/b,b/c,c/d,d/e}{
    \draw[->,>=latex] (\x) to (\y);
  }
  \foreach \x/\y/\z/\w/\sh in {e/f/{1st pass~~~~}/above/{-1},e/g/{2nd pass}/below/{1}}{
    \draw[->,>=latex] (\x) to node[midway,sloped,\w,yshift=\sh pt]{\tiny\z}(\y);
  }
  \draw[->,>=latex] (f) to [out=175,in=180](a);
 \end{tikzpicture}
 \caption{Flow chart of our shrink heuristic.}
 \label{fig:shrink}
\end{figure}
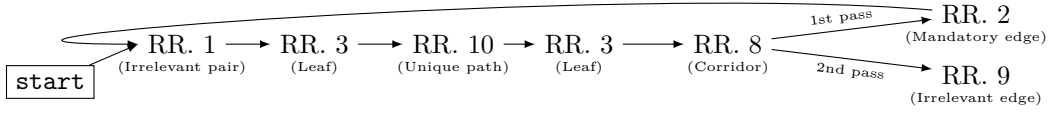
Empirically,
this two-round execution
provides a good trade-off between running time and preprocessing effectiveness
compared to
an exhaustive application until every reduction rule becomes inapplicable.
Before each of \cref{rr:leafpruning,rr:corridor},
we recompute a minimum feedback edge set.
Notably,
after applying the shrink heuristic,
the treewidth of the obtained graph is not larger
than the treewidth of the input graph,
since deletion and edge contraction and compression do not increase the treewidth.

\subsection{TD-separator Cuts}

Let $(\calT,\{B_i\}_{i\in I})$
with tree~$\calT=(I,Y)$ be a tree decomposition of~$G$.
Our cut generation is based on the following,
simple insight about cuts consisting only of unsafe edges.

\begin{observation}
\label{obs:relevantcuts}
Let $\{s_r,t_r\}\in\TPset$ be a terminal pair
and
$U\subseteq \VG(G_r^\alpha)$ such that
$|\{s_r,t_r\}\cap U| = 1$.
If $\emptyset \ne \cut_{G_r^\alpha}(U) \subseteq \Eu$,
then $|F\cap \cut_{G_r^\alpha}(U)|\geq 1$
for every solution~$F$.
\end{observation}

\begin{proof}
Since $\cut_{G_r^\alpha}(U)$ forms an $s_r$-$t_r$ cut in~$G_r^\alpha$,
it intersects all $s_r$-$t_r$ paths of detour factor at most~$\alpha$ in~$G$.
Thus,
every feasible $s_r$-$t_r$ path must traverse an edge from~$\cut_{G_r^\alpha}(U)$.
Since $\cut_{G_r^\alpha}(U)$ contains only unsafe edges,
every solution must upgrade at least one edge from~$\cut_{G_r^\alpha}(U)$.
\end{proof}

It is well known that every edge~$\{i,j\}\in Y$ of~$\calT$ corresponds to a separator~$S_{\{i,j\}} \ceq B_i \cap B_j$
in the following sense.
Let~$I_i^{\{i,j\}}$ and $I_j^{\{i,j\}}$ denote all nodes of the connected component of~$\calT-\{\{i,j\}\}$ containing~$i$ and~$j$,
respectively.
Let~$B_x^{\{i,j\}}=\bigcup_{y\in I_x^{\{i,j\}}} B_y$ denote the corresponding set of vertices in~$G$ for each~$x\in\{i,j\}$.
Then $B_i^{\{i,j\}}\setminus S_{\{i,j\}}$ is disconnected from~$B_j^{\{i,j\}}\setminus S_{\{i,j\}}$ in~$G-S_{\{i,j\}}$.
Now assume that~$\calT$ is rooted at some node.
For an edge~$\{i,j\} \in Y$,
let~$\min_\calT(\{i,j\})$ denote the endpoint of larger depth.
Every~$X\subseteq S_{\{i,j\}}$
defines a \emph{candidate side}
$V_{\{i,j\}}^X \ceq \left(B_{\min_\calT(\{i,j\})}^{\{i,j\}}\setminus S_{\{i,j\}}\right) \cup X$.
Altogether,
we consider the set~$\calV=\{V_{\{i,j\}}^X \mid \{i,j\}\in Y,\, X\subseteq S_{\{i,j\}}\}$ of all candidate sides.

For each terminal pair~$\{s_r,t_r\}$
and candidate side~$V_{\{i,j\}}^X\in \calV$,
let $V_{\{i,j\},r}^X\ceq V_{\{i,j\}}^X \cap \VG(G_r^\alpha)$
be the candidate side when restricted to the vertices of the $\alpha$-admissible graph for~$\{s_r,t_r\}$.
If $|\{s_r,t_r\}\cap V_{\{i,j\},r}^X| = 1$ and $\emptyset \ne \cut_{G_r^\alpha}(V_{\{i,j\},r}^X) \subseteq \Eu$,
then,
due to \cref{obs:relevantcuts},
every feasible solution upgrades at least one edge of this cut.
Hence,
for this terminal pair and candidate side,
we add the valid inequality
$\displaystyle
\sum\nolimits_{e\in \cut_{G_r^\alpha}(V_{\{i,j\},r}^X)} x_e \ge 1
$
to the ILP.

\section{Experiments}

Our experiments ran on
Intel\textsuperscript{\textregistered} Xeon\textsuperscript{\textregistered} Silver 4310 CPU at 2.10GHz (12 cores),
125 GB RAM,
Ubuntu 24.04.4 LTS (x86\_64), %
using Gurobi(py) 13.0 for the ILPs.

\subsection{Instances}

We considered 30 village networks and 30 region networks.
For each network,
we combined three safety models,
three terminal-pair fractions,
two random samples,
and three detour factors.
In total,
we constructed 1620 village instances and 1620 region instances.

\subparagraph*{Networks.}
We select municipalities based on the \emph{ADFC Bicycle Climate Test} \cite{adfc2024fahrradklimatest},
a nationwide online cyclist
satisfaction survey in which participants rate the perceived
bicycle-friendliness of their municipality.
From the ranking for towns of at most 20000 inhabitants,
we took the subranking when filtered for towns of at most 10000 inhabitants.
From this,
we selected the top ten,
bottom ten,
and ten randomly sampled from the remaining towns.
From OSM,
we extracted the road network using type ``drive'',
made it undirected,
kept only the largest connected component,
and
assigned the length in accordance with the OSM data.
Herein,
when two vertices are connected by two arcs,
we only represent them separately
(via paths with two edges)
if they differ significantly in length or type.
We set the cost of an unsafe edge to its length,
assuming the cost to scale with length,
thereby neglecting the type of the street segment and any further details.
This abstracts from realistic costs,
for which we did not find reliable comparable data.
Analogously,
for each of the selected towns,
we computed the network corresponding to the \vkm{3} radius region around the associated town's center,
mimicking inter-village networks.
\cref{tab:village:struct} gives an overview of selected statistics relevant to us
for the derived networks.
\begin{table}
  \centering
  \begin{tabular}{lrrrrrr}
  \toprule
  Type & Vertices & Edges~$m$ & $\fes/m$ & $\twFi$ & $\Delta$ & Planar \\ %
  \midrule
  Village
  & 355.43 \taban{341}{61}{655} %
  & 443.67 \taban{412.50}{76}{871}
  & 0.2 \taban{0.20}{0.11}{0.25}
  & 6.00 \taban{5}{3}{11}
  & 4.17 \taban{4}{4}{5}
  & 27/30
  \\
  Region
  & 738.63 \taban{604.50}{140}{1929}
  & 934.97 \taban{755}{176}{2623}
  & 0.2 \taban{0.20}{0.14}{0.28}
  & 7.97 \taban{7}{4}{14}
  & 4.43 \taban{4}{4}{5}
  & 20/30
  \\
  \bottomrule
  \end{tabular}
  \caption{Structural statistics for our village and region networks as to the means with median, minimum (lower), and maximum (upper) in parentheses.
  $\twFi$ denotes the width of a TD computed via minimum fill-in.}
  \label{tab:village:struct}
\end{table}
Remarkably,
on average,
20\% of the edges suffice as a feedback edge set,
and the maximum observed treewidth upper bound is 11 for villages and 14 for regions.
We also point out that only two thirds of our region networks are planar.

\subparagraph*{Safety models.}
For classifying edges as unsafe,
we developed three models A, B, and C
that are nested in spirit.
Model A is an accommodating local-access model where
local
roads are generally considered bicycle-suitable.
Model B is a balanced speed-aware model where
local roads require low-speed evidence.
Finally,
model C is a conservative low-stress model where
only clearly low-stress roads are considered safe.
See~\cref{tab:safety-models} for further details.
\begin{table}[t]
  \centering
  \begin{tblr}{
    width = \textwidth,
    colspec = {X[1.10,l] X[1.00,l] X[1.00,l] X[1.55,l]},
    cells = {valign=m},
    row{1} = {halign=l},
    hspan = minimal,
    }
    \toprule
    & Model A & Model B & Model C \\
    \midrule
    Safe in all models
    &
    \SetCell[c=3]{l,bg=gray!10}
    Living streets, traffic-calmed roads, and extracted road segments with explicit bicycle infrastructure
    &
    &
    \\

    \midrule
    Residential roads
    &
    \SetCell[r=2]{c,bg=gray!10} Safe
    &
    \SetCell[r=2]{l,bg=gray!5} Safe if speed is missing or at most \vkmh{30}
    &
    Safe if speed is missing or at most \vkmh{30}
    \\

    Service / unclassified roads
    &
    &
    &
    \SetCell{l,bg=gray!10} Safe only if speed is at most \vkmh{20}
    \\

    Tertiary roads
    &
    \SetCell[c=2]{l,bg=gray!15} Safe only if speed is at most \vkmh{30}
    &
    &
    Unsafe by default
    \\

    Paths / pedestrian streets, if present
    &
    \SetCell[c=2]{l,bg=gray!10} Safe unless bicycles are explicitly forbidden
    &
    &
    Safe only if bicycles are explicitly allowed
    \\

    \midrule
    \SetCell[r=2]{l} Average unsafe fraction
    &
    $0.22$ \taban{0.21}{0.10}{0.32}
    &
    $0.29$ \taban{0.28}{0.12}{0.51}
    &
    $0.37$ \taban{0.36}{0.15}{0.58}
    \\
    &
    $0.22$ \taban{0.22}{0.14}{0.30}
    &
    $0.30$ \taban{0.28}{0.15}{0.56}
    &
    $0.36$ \taban{0.35}{0.15}{0.57}
    \\
    \bottomrule
  \end{tblr}
  \caption{
  Summary of the three edge safety models.
  The first row gives rules that make an edge safe in all models;
  subsequent rows list additional rules specific to road classes.
  Last two rows report average unsafe-edge fractions, with median, minimum, and maximum in parentheses, for villages (top row) and regions (bottom row).
  }
  \label{tab:safety-models}
\end{table}

\subparagraph*{Terminal pairs and detour-factors.}
For each network,
we randomly sampled~$x\cdot n$ many (distinct) terminal pairs
for~$x\in\{0.25,0.5,0.75\}$ for villages and for~$x\in\{0.15,0.3,0.45\}$ for regions,
where $n$ denotes the number of vertices.
For each network and such setup,
we constructed two instances.
Finally,
for each of those,
we chose detour factors~$\alpha\in\{1.2,1.3,1.5\}$,
where~$\alpha=1.2$~\cite[p.~10]{FGSV2010ERA} and~$\alpha=1.3$~\cite{CaiOM2024} are
used in bicycle-network design as target bounds.

\subsection{Algorithmic Setup}
\label{ssec:algsetup}

We compare \ILP{} (cf.~\eqref{eq:ilp-base})
against \ILPcut{}, which is \ILP{} with TD-separator cuts,
\ILPp{}, which is \ILP{} where preprocessing is applied,
and \ILPpcut{},
the combination of both.
To compute a tree decomposition,
we use NetworkX's
implementation
\texttt{treewidth\_min\_fill\_in} \cite{hagberg2008networkx}
of the minimum fill-in heuristic \cite{bodlaender2010treewidth,rose1976algorithmic}.
On a high level,
the heuristic greedily turns neighborhoods into cliques,
thereby creating for the resulting graph a so-called perfect elimination ordering,
which yields a tree decomposition.
The number of candidate sides for each TD-separator
can be as large
as~$2^{\twFi+1}$.
Each such candidate can be considered for every terminal pair,
leading to many candidate cuts.
Thus,
our implementation limits the process of generating TD-separator cuts
as part of the ILP building time in the following three ways.
\begin{enumerate}
 \item We skip TD-separators of size larger than~$s^*$,
 since enumerating all $2^{s^*}$ candidate sets can be expensive.
 We set~$s^*=14$;
 recall that the maximum~$\twFi$ in our experiments is at most 14,
 implying that the limit applies only to boundary cases in our experiments.
 \item We only consider the best $N^*$ candidate sides for cut generation according to the following ranking.
 Let~$U\in \calV$ be a candidate side with~$\cut_G(U)\neq\emptyset$.
 Define the (negative) fraction~$\rho_G(U)=-|\cut_G(U)\cap \Eu|/|\cut_G(U)|$ of unsafe edges in the cut.
 The score of~$U$ in~$G$ is defined as $\operatorname{score}(U)=(|\cut_G(U)|,\rho_G(U),|U|)$.
 Each dimension of the score is to minimize,
 and scores are compared lexicographically.
 That is,
 firstly we seek small cuts,
 secondly high fractions of unsafe edges,
 and thirdly small sides.
 Ties are broken arbitrarily
 (by their time of appearance)
 for any two candidate sides with the same score.
 We set~$N^*=5000$.
 \item For each terminal pair~$\{s_i,t_i\}$,
 we only consider cuts with at most $m^*$ edges.
 From these,
 we only add at most~$c^*$ cuts to the ILP.
 We set~$m^*=10$ and~$c^*=20$.
\end{enumerate}
These limits affect only \emph{which} valid cuts are generated;
every added inequality is valid.

\subsection{Results}

\subparagraph*{Solver performance.}
Overall,
from our experiments,
we can conclude the following
regarding the solver performance in terms of runtimes;
\cref{fig:solver_performance} and \cref{fig:scores}
accompany our conclusions.%
\begin{figure}[t]
 \includegraphics[width=0.3284\textwidth]{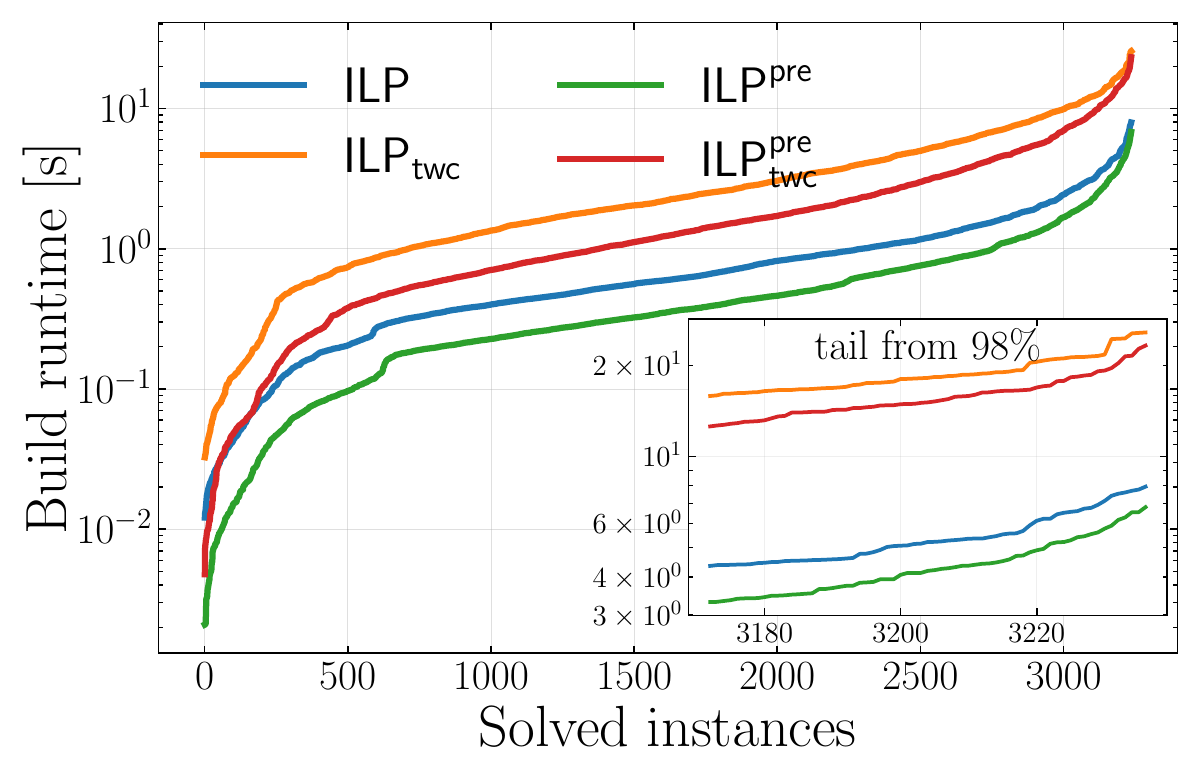}
 \hfill
 \includegraphics[width=0.3284\textwidth]{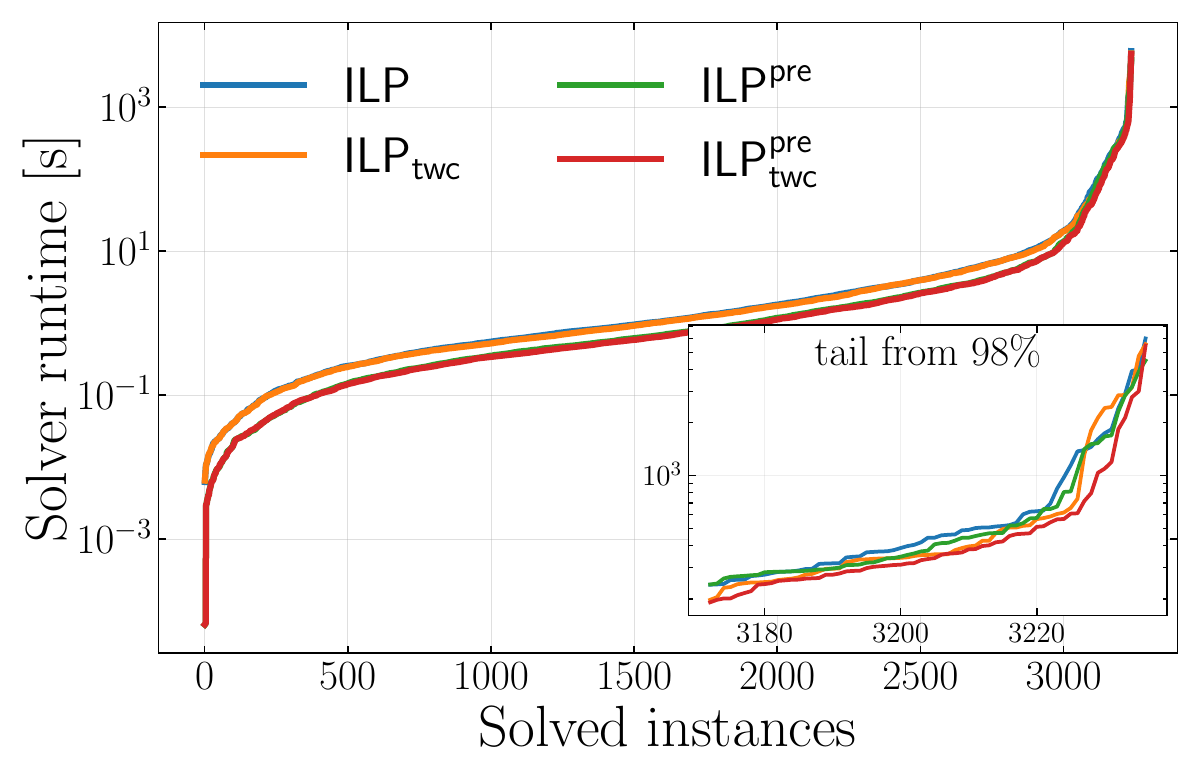}
 \hfill
 \includegraphics[width=0.3284\textwidth]{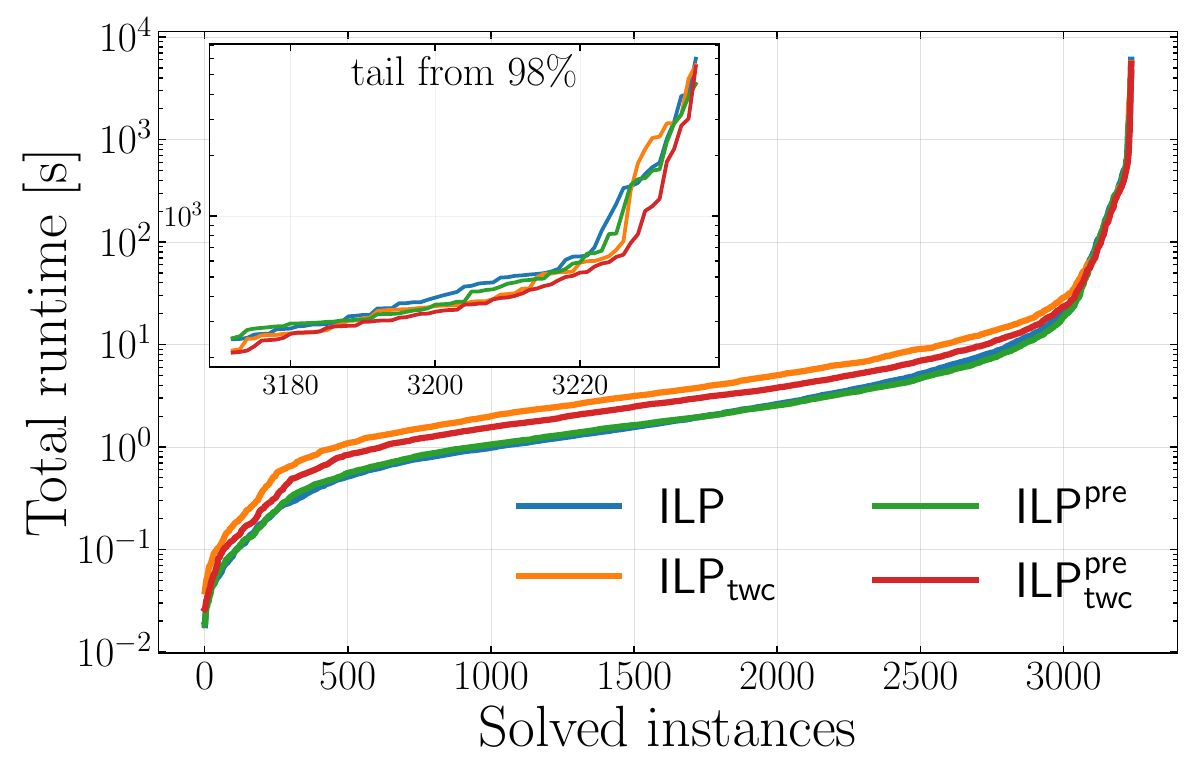}
 \caption{Cactus plots for all instances solved by each of the four solvers
 regarding building time (left),
 solver runtime (middle),
 and total runtime (right).
 Each inset zooms into the tail of the corresponding cactus plot,
 i.e.,
 the slowest 2\% of the common solved instances.}
 \label{fig:solver_performance}
\end{figure}
\begin{enumerate}
 \item Preprocessing significantly decreases the solver runtime;
 see \cref{fig:solver_performance} (middle) and \cref{fig:scores}.
 Due to the short runtimes of the preprocessing,
 both solvers suing preprocessing \ILPp{} and \ILPpcut{} score best versus~\ILP;
 see \cref{fig:scores}.
 \item The runtime for building the ILP with the TD-separator cuts but without preprocessing is \vsec{3.53} on average and never exceeded \vsec{25.64};
 see \cref{fig:solver_performance} (left).
 This indicates that the heuristic
 makes use of small treewidth,
 small degree,
 and
 sparsity.
 Preprocessing reduced the building times for both with and without TD-separator cut generation,
 to roughly 65\% for the average time and 90\% of the maximum observed time.
 \item For instances solved within roughly~\vsec{100},
 TD-separator cuts do not appear to pay off.
 For harder instances,
 however,
 \ILPpcut{} dominates;
 see \cref{fig:solver_performance} (right) and \cref{fig:scores}.
 \item While \ILPcut{} performs worse than \ILP{} in total,
 its solver runtime is mostly shorter;
 see \cref{fig:scores} (left).
 While preprocessing also decreases the solver runtime,
 finally \ILPpcut{} scores best against all other solvers,
 even against \ILPp{} when it comes to solver runtime.
\end{enumerate}

\begin{figure}[t]
 \includegraphics[width=0.495\textwidth]{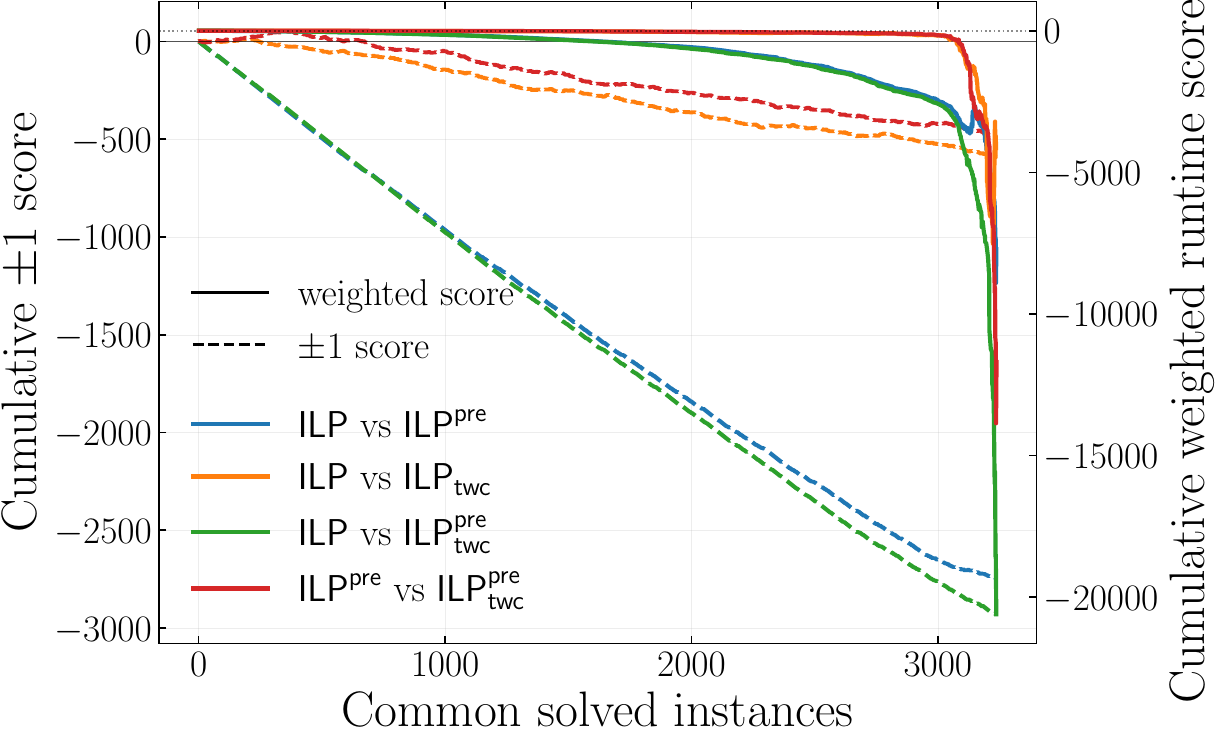}
 \hfill
 \includegraphics[width=0.495\textwidth]{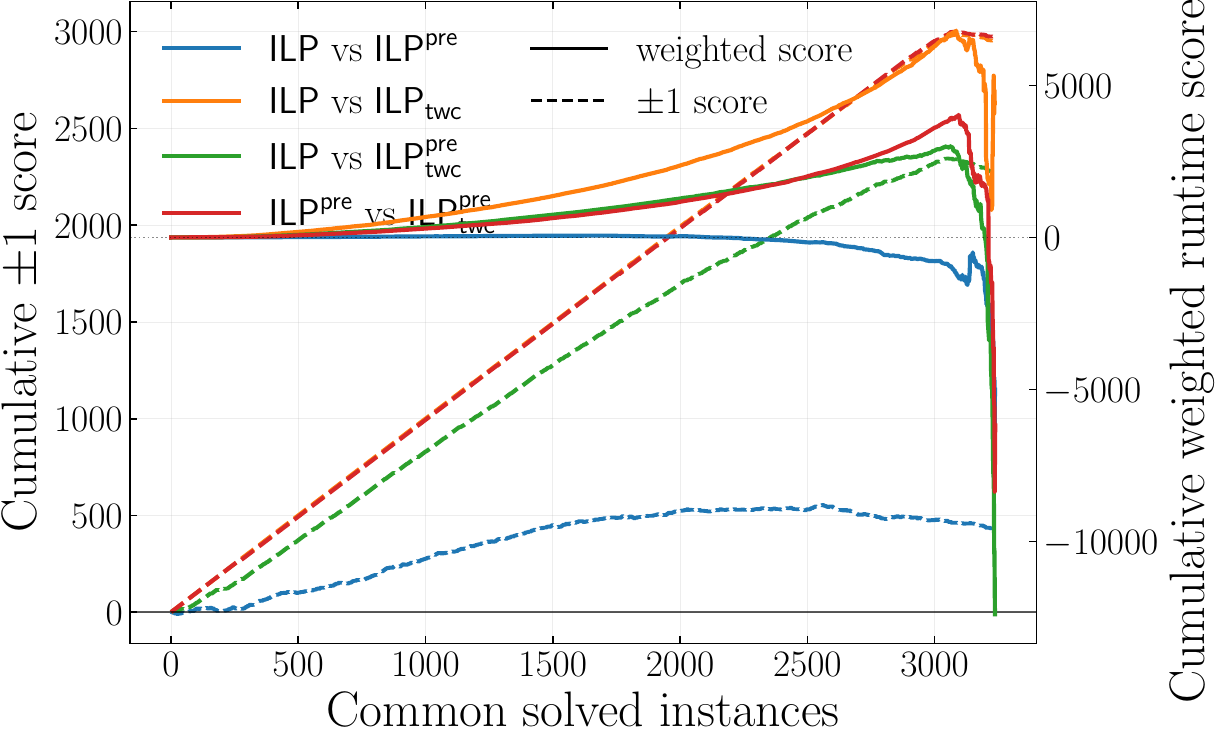}
 \caption{$\pm1$- and weighted scores of solver pairs $\calS_1$ versus $\calS_2$.
 The weighted score is $s(\calS_1,\calS_2)=\operatorname{time}(\calS_2) - \operatorname{time}(\calS_1)$,
 where $\operatorname{time}$ is either the solver runtime (left) or the total runtime (right).
 The $\pm1$-score is +1 if $s(\calS_1,\calS_2)>0$,
 and -1 if $s(\calS_1,\calS_2)<0$.
 So,
 a line going upwards corresponds to $\calS_1$ is faster than $\calS_2$,
 and downwards to $\calS_1$ is slower than $\calS_2$.
 For each solver pair,
 the commonly solved instances are sorted separately.
 }
 \label{fig:scores}
\end{figure}

\subparagraph*{Effect of the preprocessing.}
Our preprocessing routine (cf.~\cref{sec:fesbasedpreprocessing})
is fast and effective.
Over all of our instances,
the routine takes $1.15$ seconds on average (median $0.67$ second),
where the maximum observed time is $12.46$ seconds.
Yet,
the routine significantly reduces
the number of unsafe edges and of terminal pairs  to
63\% and 73\% on average,
respectively.
See \cref{fig:kernel_effect} for a detailed overview.
\begin{figure}[t]
 \includegraphics[width=0.99\textwidth]{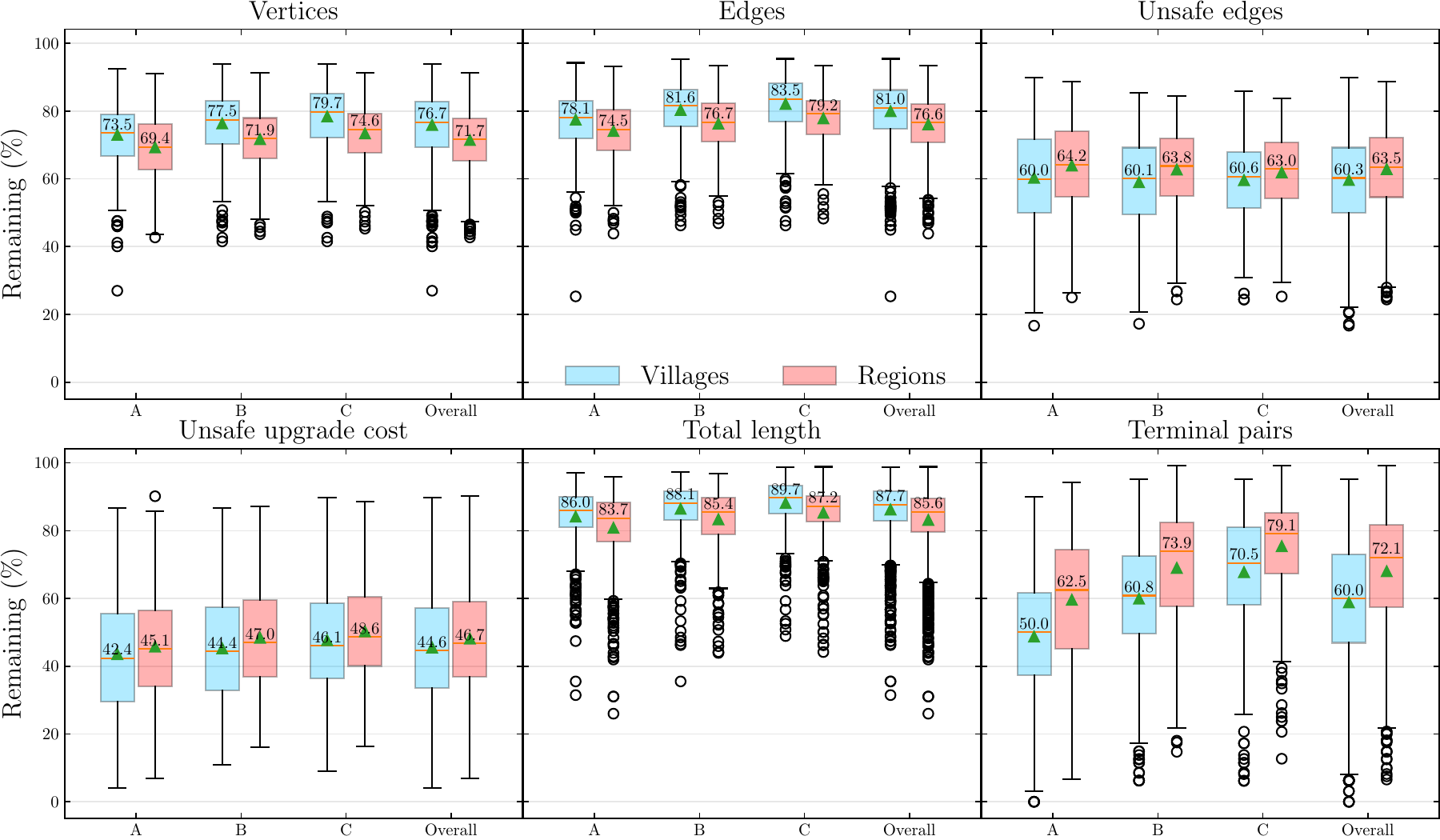}
 \caption{Statistics for the preprocessing routine.}
 \label{fig:kernel_effect}
\end{figure}
For our three different safety models A, B, and C,
we find only a correlation regarding the reduction of terminal pairs.
Comparing safety models A and C,
the reduction changes from 50\% to 70\% for villages,
and 62\% to 79\% for regions.

\subparagraph*{Cost of detours.}
\cref{tab:costdetour} gives an overview of the costs when the detour factor is increased.
\begin{table}[t]
 \centering
\begin{tabular}{lrrr}
\toprule
 & $1.2\to1.3$ & $1.3\to1.5$ & $1.2\to1.5$ \\
\midrule
Cost reduction (in \%) & $7.87$ \taban{6.59}{0.00}{37.86} & $10.76$ \taban{8.97}{0.00}{39.21} & $17.63$ \taban{15.42}{0.00}{57.29} \\
\bottomrule
\end{tabular}
\caption{Overview of the average relative cost reduction when increasing the detour factor~$\alpha$, with median, minimum (bottom), and maximum (top) in parentheses.}
\label{tab:costdetour}
\end{table}
Recall that we set the cost of an unsafe edge to its length;
hence,
our results display the savings of the total upgraded length.
We highlight that when moving from $\alpha=1.2$ to $\alpha=1.3$,
on average 7.87\% of the costs are reduced,
ranging from no reduction even up to 37.86\% in the maximum case.
For decision makers,
we can often assume that the upper bound on the detour factor lies within a reasonable interval,
say~$\alpha\in[1.2,1.3]$.
Our results hence indicate that
it is worth computing solutions for several upper bounds on the detour factor in the given interval
and compare the trade-offs
between detour length and budget.

\section{Discussion}

On the experimental-modeling side,
we randomly sample terminal pairs,
whereas real-world demand is likely more structured around hotspots
such as schools, supermarkets, or public-transport stops.
Similarly,
we set the upgrade cost of each street segment to its length.
Besides length,
however,
other factors such as road type or local regulations may affect the construction costs.
Incorporating such data would make the experiments more realistic.

On the algorithmic side,
our experiments show that preprocessing substantially reduces solver runtimes,
making further and improved reduction rules a promising direction.
For TD-separator cuts,
their effectiveness
may depend strongly on the generated candidate sides.
Larger heuristic parameters increase model-building time,
and adding too many cuts can even slow down the ILP solver
(see \cref{sec:app:fer} for details).
A better understanding of how to select strong cuts and which candidate sides yield them
is an important future direction.

On the theoretical side,
it remains open which structural parameters beyond $\nue$ and $\fes+\ntp$
lead to fixed-parameter tractability,
in particular whether $\fvs+\ntp$ or the feedback edge number alone is sufficient.

\newpage
\bibliography{bineim-bib}

@inproceedings{hagberg2008networkx,
  author    = {Hagberg, Aric A. and Schult, Daniel A. and Swart, Pieter J.},
  title     = {Exploring Network Structure, Dynamics, and Function using NetworkX},
  booktitle = {Proceedings of the 7th Python in Science Conference},
  editor    = {Varoquaux, Ga{\"e}l and Vaught, Travis and Millman, Jarrod},
  pages     = {11--15},
  address   = {Pasadena, CA USA},
  year      = {2008}
}

@incollection{Karp1972,
  author    = {Karp, Richard M.},
  title     = {Reducibility Among Combinatorial Problems},
  booktitle = {Complexity of Computer Computations},
  editor    = {Miller, Raymond E. and Thatcher, James W. and Bohlinger, Jean D.},
  pages     = {85--103},
  publisher = {Plenum Press},
  address   = {New York},
  year      = {1972},
  doi       = {10.1007/978-3-540-68279-0_8}
}

@article{ImpagliazzoPaturi2001,
  author  = {Impagliazzo, Russell and Paturi, Ramamohan},
  title   = {On the Complexity of k-SAT},
  journal = {Journal of Computer and System Sciences},
  volume  = {62},
  number  = {2},
  pages   = {367--375},
  year    = {2001},
  doi     = {10.1006/jcss.2000.1727}
}

@article{ImpagliazzoPZ2001,
  author  = {Impagliazzo, Russell and Paturi, Ramamohan and Zane, Francis},
  title   = {Which Problems Have Strongly Exponential Complexity?},
  journal = {Journal of Computer and System Sciences},
  volume  = {63},
  number  = {4},
  pages   = {512--530},
  year    = {2001},
  doi     = {10.1006/jcss.2001.1774}
}

@article{KellerhalsKoana2022,
  author  = {Kellerhals, Leon and Koana, Tomohiro},
  title   = {Parameterized Complexity of Geodetic Set},
  journal = {Journal of Graph Algorithms and Applications},
  volume  = {26},
  number  = {4},
  pages   = {401--419},
  year    = {2022},
  doi     = {10.7155/jgaa.00601},
  url     = {https://doi.org/10.7155/jgaa.00601}
}

@book{cygan2015parameterized,
  author =        {Cygan, Marek and Fomin, Fedor V. and
                   Kowalik, {\L}ukasz and Lokshtanov, Daniel and
                   Marx, D{\'a}niel and Pilipczuk, Marcin and
                   Pilipczuk, Micha{\l} and Saurabh, Saket},
  publisher =     {Springer},
  title =         {Parameterized Algorithms},
  year =          {2015},
}

@article{CaiOM2024,
  author  = {Cai, Yutong and Ong, Ghim Ping and Meng, Qiang},
  title   = {Sidewalk-Based Bicycle Path Network Design Incorporating Equity in Cycling Time},
  journal = {Computer-Aided Civil and Infrastructure Engineering},
  volume  = {39},
  number  = {20},
  pages   = {3063--3082},
  year    = {2024},
  doi     = {10.1111/mice.13240},
  publisher = {Wiley}
}

@article{bodlaender2010treewidth,
  author  = {Hans L. Bodlaender and Arie M. C. A. Koster},
  title   = {Treewidth computations I. Upper bounds},
  journal = {Information and Computation},
  volume  = {208},
  number  = {3},
  pages   = {259--275},
  year    = {2010},
  doi     = {10.1016/j.ic.2009.03.008}
}

@article{rose1976algorithmic,
  author  = {Donald J. Rose and R. Endre Tarjan and George S. Lueker},
  title   = {Algorithmic Aspects of Vertex Elimination on Graphs},
  journal = {SIAM Journal on Computing},
  volume  = {5},
  number  = {2},
  pages   = {266--283},
  year    = {1976},
  doi     = {10.1137/0205021}
}

@book{FGSV2010ERA,
  author    = {{Forschungsgesellschaft f{\"u}r Stra{\ss}en- und Verkehrswesen}},
  title     = {{Empfehlungen f{\"u}r Radverkehrsanlagen: ERA}},
  series    = {FGSV 284},
  edition   = {{A}usgabe 2010},
  publisher = {FGSV Verlag},
  address   = {K{\"o}ln},
  year      = {2010},
  isbn      = {978-3-941790-63-6},
  url       = {https://www.fgsv-verlag.de/era},
  note      = {See p.~10 for network-level detour-factor recommendations}
}

@techreport{DollBrauerDuffner2024,
  author      = {Doll, Claus and Brauer, Clemens and Duffner-Korbee, Dorien},
  title       = {The Potential of Cycling for Climate Protection and Livable Urban Centers and Regions:
                 New Methods for Forecasting Supply and Demand in Germany as a Cycling Nation up to 2035},
  institution = {Fraunhofer Institute for Systems and Innovation Research ISI},
  address      = {Karlsruhe, Germany},
  year         = {2024},
  month        = may,
  date         = {2024-05-21},
  note         = {Prepared on behalf of the Allgemeiner Deutscher Fahrrad-Club (ADFC)},
  url          = {https://www.adfc.de/fileadmin/user_upload/Expertenbereich/Politik_und_Verwaltung/Studie_Potenzial_des_Radverkehrs_in_Deutschland/Cycling_Potentials_Germany_Summary_Report.pdf}
}

@misc{adfc2024fahrradklimatest,
  author       = {{Allgemeiner Deutscher Fahrrad-Club e.V. (ADFC)}},
  title        = {{ADFC-Fahrradklima-Test 2024: Rankingliste}},
  year         = {2025},
  howpublished = {\url{https://fahrradklima-test.adfc.de/fileadmin/BV/FKT/Download-Material/Ergebnisse_2024/Download-Element/Rankingliste_FKT_2024.pdf}},
  note         = {Accessed: 2026-06-17}
}

@article{VanDuelmenSK2022,
title = {Transport poverty meets car dependency: A GPS tracking study of socially disadvantaged groups in European rural peripheries},
journal = {Journal of Transport Geography},
volume = {101},
pages = {103351},
year = {2022},
issn = {0966-6923},
doi = {https://doi.org/10.1016/j.jtrangeo.2022.103351},
url = {https://www.sciencedirect.com/science/article/pii/S0966692322000746},
author = {Christoph {van Dülmen} and Martin Šimon and Andreas Klärner}
}

@article{PoltimaeRRP2022,
  author  = {Poltimäe, Helen and Rehema, Merlin and Raun, Janika and Poom, Age},
  title   = {In search of sustainable and inclusive mobility solutions for rural areas},
  journal = {European Transport Research Review},
  year    = {2022},
  volume   = {14},
  number   = {1},
  pages    = {13},
  doi      = {10.1186/s12544-022-00536-3},
  url      = {https://doi.org/10.1186/s12544-022-00536-3},
  issn     = {1866-8887}
}

@article{MuellerRCNDGGIKN2015,
title = {Health impact assessment of active transportation: A systematic review},
journal = {Preventive Medicine},
volume = {76},
pages = {103-114},
year = {2015},
issn = {0091-7435},
doi = {https://doi.org/10.1016/j.ypmed.2015.04.010},
url = {https://www.sciencedirect.com/science/article/pii/S0091743515001164},
author = {Natalie Mueller and David Rojas-Rueda and Tom Cole-Hunter and Audrey {de Nazelle} and Evi Dons and Regine Gerike and Thomas Götschi and Luc {Int Panis} and Sonja Kahlmeier and Mark Nieuwenhuijsen}
}

@article{OjaTBGKRK2011,
  author    = {Oja, Pekka and Titze, Sylvia and Bauman, Adrian and de Geus, Bas and Krenn, Patrick and Reger-Nash, Bill and Kohlberger, Thomas},
  title     = {Health Benefits of Cycling: A Systematic Review},
  journal   = {Scandinavian Journal of Medicine \& Science in Sports},
  year      = {2011},
  volume     = {21},
  number     = {4},
  pages      = {496--509},
  month      = aug,
  doi        = {10.1111/j.1600-0838.2011.01299.x},
  pmid       = {21496106},
  note       = {Epub 2011 Apr 18}
}

@book{OECD1984,
  author = {{European Conference of Ministers of Transport}},
  title = {Public Transport in Rural Areas},
  publisher = {OECD Publishing},
  year = {1984},
  doi = {10.1787/9789282105252-en}
}

@article{VieroSzell2025,
   title={Network Analysis of the Danish Bicycle Infrastructure: Bikeability Across Urban–Rural Divides},
   volume={57},
   ISSN={1538-4632},
   url={http://dx.doi.org/10.1111/gean.70012},
   DOI={10.1111/gean.70012},
   number={4},
   journal={Geographical Analysis},
   publisher={Wiley},
   author={Vierø, Ane Rahbek and Szell, Michael},
   year={2025},
   month=May, 
   pages={616–640} 
}

@article{GorgesMingardo2025,
title = {The potential of active modes to reduce short car trips. A data-driven approach},
journal = {Transport Policy},
volume = {168},
pages = {1-14},
year = {2025},
issn = {0967-070X},
doi = {https://doi.org/10.1016/j.tranpol.2025.03.015},
url = {https://www.sciencedirect.com/science/article/pii/S0967070X25001106},
author = {Hannah Gorges and Giuliano Mingardo}
}

@article{BeckxBDBI2013,
title = {Limits to active transport substitution of short car trips},
journal = {Transportation Research Part D: Transport and Environment},
volume = {22},
pages = {10-13},
year = {2013},
issn = {1361-9209},
doi = {https://doi.org/10.1016/j.trd.2013.03.001},
url = {https://www.sciencedirect.com/science/article/pii/S1361920913000424},
author = {Carolien Beckx and Steven Broekx and Bart Degraeuwe and Bart Beusen and Luc {Int Panis}}
}

@techreport{FahrradMonitor2025,
  author      = {{SINUS-Institut} and {Bundesministerium für Verkehr}},
  title       = {{Fahrrad-Monitor Deutschland 2025 -- Erwachsenen- und Jugendbefragung (14--69 Jahre)}},
  institution = {Bundesministerium für Verkehr},
  year        = {2026},
  url          = {https://www.bmv.de/SharedDocs/DE/Anlage/StV/fahrrad-monitor-deutschland-2025-langfassung-erwachsenen-jugendbefragung.pdf}
}

@article{DomLS14,
  author       = {Michael Dom and
                  Daniel Lokshtanov and
                  Saket Saurabh},
  title        = {Kernelization Lower Bounds Through Colors and IDs},
  journal      = {{ACM} Trans. Algorithms},
  volume       = {11},
  number       = {2},
  pages        = {13:1--13:20},
  year         = {2014},
  url          = {https://doi.org/10.1145/2650261},
  doi          = {10.1145/2650261},
  bibsource    = {dblp computer science bibliography, https://dblp.org}
}

@article{CyganDLMNOPSW16,
  author       = {Marek Cygan and
                  Holger Dell and
                  Daniel Lokshtanov and
                  D{\'{a}}niel Marx and
                  Jesper Nederlof and
                  Yoshio Okamoto and
                  Ramamohan Paturi and
                  Saket Saurabh and
                  Magnus Wahlstr{\"{o}}m},
  title        = {On Problems as Hard as {CNF-SAT}},
  journal      = {{ACM} Trans. Algorithms},
  volume       = {12},
  number       = {3},
  pages        = {41:1--41:24},
  year         = {2016},
  url          = {https://doi.org/10.1145/2925416},
  doi          = {10.1145/2925416},
  bibsource    = {dblp computer science bibliography, https://dblp.org}
}

@article{LimDGH22,
author = {Jisoon Lim  and Kevin Dalmeijer  and Subhrajit Guhathakurta  and Pascal Van Hentenryck },
title = {The Bicycle Network Improvement Problem},
journal = {Journal of Transportation Engineering, Part A: Systems},
volume = {148},
number = {11},
pages = {04022095},
year = {2022},
doi = {10.1061/JTEPBS.0000742},

URL = {https://ascelibrary.org/doi/abs/10.1061/JTEPBS.0000742},
eprint = {https://ascelibrary.org/doi/pdf/10.1061/JTEPBS.0000742}
}

@article{DuthieUnnikrishnan2014,
  author  = {Duthie, Jennifer and Unnikrishnan, Avinash},
  title   = {Optimization Framework for Bicycle Network Design},
  journal = {Journal of Transportation Engineering},
  volume  = {140},
  number  = {7},
  pages   = {04014028},
  year    = {2014},
  doi     = {10.1061/(ASCE)TE.1943-5436.0000690},
  url     = {https://doi.org/10.1061/(ASCE)TE.1943-5436.0000690}
}

@article{MauttoneMRT2017,
  author  = {Mauttone, Antonio and Mercadante, Gonzalo and Rabaza, Mar{\'i}a and Toledo, Fernanda},
  title   = {Bicycle Network Design: Model and Solution Algorithm},
  journal = {Transportation Research Procedia},
  volume  = {27},
  pages   = {969--976},
  year    = {2017},
  doi     = {10.1016/j.trpro.2017.12.119}
}

@article{NateraOBIS2020,
  author  = {Natera Orozco, Luis Guillermo and Battiston, Federico and I{\~n}iguez, Gerardo and Szell, Michael},
  title   = {Data-driven strategies for optimal bicycle network growth},
  journal = {Royal Society Open Science},
  volume  = {7},
  number  = {12},
  pages   = {201130},
  year    = {2020},
  doi     = {10.1098/rsos.201130},
  url     = {https://doi.org/10.1098/rsos.201130}
}

@article{FurthMN2016,
  author  = {Furth, Peter G. and Mekuria, Maaza C. and Nixon, Hilary},
  title   = {Network Connectivity for Low-Stress Bicycling},
  journal = {Transportation Research Record: Journal of the Transportation Research Board},
  volume  = {2587},
  number  = {1},
  pages   = {41--49},
  year    = {2016},
  doi     = {10.3141/2587-06},
  url     = {https://doi.org/10.3141/2587-06}
}

@article{SchonerLevinson2014,
  author  = {Schoner, Jessica E. and Levinson, David M.},
  title   = {The Missing Link: Bicycle Infrastructure Networks and Ridership in 74 {US} Cities},
  journal = {Transportation},
  volume  = {41},
  number  = {6},
  pages   = {1187--1204},
  year    = {2014},
  doi     = {10.1007/s11116-014-9538-1},
  url     = {https://doi.org/10.1007/s11116-014-9538-1}
}

@article{SzellMPGS2022,
  author  = {Szell, Michael and Mimar, Sayat and Perlman, Tyler and Ghoshal, Gourab and Sinatra, Roberta},
  title   = {Growing Urban Bicycle Networks},
  journal = {Scientific Reports},
  volume  = {12},
  pages   = {6765},
  year    = {2022},
  doi     = {10.1038/s41598-022-10783-y},
  url     = {https://doi.org/10.1038/s41598-022-10783-y}
}

@article{SteinackerSTM2022,
  author  = {Steinacker, Christoph and Storch, David-Maximilian and Timme, Marc and Schr{\"o}der, Malte},
  title   = {Demand-driven design of bicycle infrastructure networks for improved urban bikeability},
  journal = {Nature Computational Science},
  volume  = {2},
  pages   = {655--664},
  year    = {2022},
  doi     = {10.1038/s43588-022-00318-w},
  url     = {https://doi.org/10.1038/s43588-022-00318-w}
}

@inproceedings{CyganGK2013,
  author    = {Cygan, Marek and Grandoni, Fabrizio and Kavitha, Telikepalli},
  title     = {On Pairwise Spanners},
  booktitle = {30th International Symposium on Theoretical Aspects of Computer Science (STACS 2013)},
  series    = {Leibniz International Proceedings in Informatics (LIPIcs)},
  volume    = {20},
  pages     = {209--220},
  publisher = {Schloss Dagstuhl -- Leibniz-Zentrum f{\"u}r Informatik},
  year      = {2013},
  doi       = {10.4230/LIPIcs.STACS.2013.209},
  url       = {https://doi.org/10.4230/LIPIcs.STACS.2013.209}
}

@article{ChlamtacDKL2020,
  author  = {Chlamt{\'{a}}c, Eden and Dinitz, Michael and Kortsarz, Guy and Laekhanukit, Bundit},
  title   = {Approximating Spanners and Directed Steiner Forest: Upper and Lower Bounds},
  journal = {ACM Transactions on Algorithms},
  volume  = {16},
  number  = {3},
  pages   = {33:1--33:31},
  year    = {2020},
  doi     = {10.1145/3381451},
  url     = {https://doi.org/10.1145/3381451}
}

@article{GrigorescuKL2026,
  author  = {Grigorescu, Elena and Kumar, Nithin and Lin, Young-San},
  title   = {Approximation Algorithms for Directed Weighted Spanners},
  journal = {Algorithmica},
  volume  = {88},
  number  = {3},
  pages   = {38},
  year    = {2026},
  doi     = {10.1007/s00453-026-01384-6},
  url     = {https://doi.org/10.1007/s00453-026-01384-6}
}

@article{Kobayashi2018,
  author  = {Kobayashi, Yusuke},
  title   = {{NP}-hardness and Fixed-Parameter Tractability of the Minimum Spanner Problem},
  journal = {Theoretical Computer Science},
  volume  = {746},
  pages   = {88--97},
  year    = {2018},
  doi     = {10.1016/j.tcs.2018.06.031},
  url     = {https://doi.org/10.1016/j.tcs.2018.06.031}
}

@inproceedings{Kobayashi2020,
  author       = {Yusuke Kobayashi},
  editorx       = {Christophe Paul and
                  Markus Bl{\"{a}}ser},
  title        = {An {FPT} Algorithm for Minimum Additive Spanner Problem},
  booktitle    = {Proceedings of the 37th International Symposium on Theoretical Aspects of Computer Science ({STACS} 2020)},
  series       = {LIPIcs},
  volume       = {154},
  pages        = {11:1--11:16},
  publisher    = {Schloss Dagstuhl - Leibniz-Zentrum f{\"{u}}r Informatik},
  year         = {2020},
  url          = {https://doi.org/10.4230/LIPIcs.STACS.2020.11},
  doi          = {10.4230/LIPICS.STACS.2020.11},
  bibsource    = {dblp computer science bibliography, https://dblp.org}
}

@article{FeldmannLampis2025,
  author  = {Feldmann, Andreas Emil and Lampis, Michael},
  title   = {Parameterized Algorithms for Steiner Forest in Bounded Width Graphs},
  journal = {ACM Transactions on Algorithms},
  year    = {2025},
  doi     = {10.1145/3748724},
  url     = {https://doi.org/10.1145/3748724}
}

@article{SimonovSV2026,
  author        = {Simonov, Kirill and Soheil, Farehe and Verma, Shaily},
  title         = {Finding Minimum Distance Preservers: A Parameterized Study},
  journal       = {arXiv preprint arXiv:2603.21442},
  year          = {2026},
  eprint        = {2603.21442},
  archivePrefix = {arXiv},
  primaryClass  = {cs.DS},
  url           = {https://arxiv.org/abs/2603.21442}
}

@article{GASSNER2010,
title = {The Steiner Forest Problem revisited},
journal = {Journal of Discrete Algorithms},
volume = {8},
number = {2},
pages = {154-163},
year = {2010},
note = {Selected papers from the 3rd Algorithms and Complexity in Durham Workshop ACiD 2007},
issn = {1570-8667},
doi = {https://doi.org/10.1016/j.jda.2009.05.002},
url = {https://www.sciencedirect.com/science/article/pii/S1570866709000628},
author = {Elisabeth Gassner}
}

@inproceedings{ManiuSJ19,
  author       = {Silviu Maniu and
                  Pierre Senellart and
                  Suraj Jog},
  editorx       = {Pablo Barcel{\'{o}} and
                  Marco Calautti},
  title        = {An Experimental Study of the Treewidth of Real-World Graph Data},
  booktitle    = {22nd International Conference on Database Theory, {ICDT} 2019, Lisbon,
                  Portugal, March 26-28, 2019},
  series       = {LIPIcs},
  pages        = {12:1--12:18},
  publisher    = {Schloss Dagstuhl - Leibniz-Zentrum f{\"{u}}r Informatik},
  year         = {2019},
  url          = {https://doi.org/10.4230/LIPIcs.ICDT.2019.12},
  doi          = {10.4230/LIPICS.ICDT.2019.12},
  bibsource    = {dblp computer science bibliography, https://dblp.org}
}

@article{Boeing2017,
  author  = {Boeing, Geoff},
  title   = {{OSMnx}: New Methods for Acquiring, Constructing, Analyzing, and Visualizing Complex Street Networks},
  journal = {Computers, Environment and Urban Systems},
  volume  = {65},
  pages   = {126--139},
  year    = {2017},
  doi     = {10.1016/j.compenvurbsys.2017.05.004},
  url     = {https://doi.org/10.1016/j.compenvurbsys.2017.05.004}
}

@article{Boeing2020,
  author  = {Boeing, Geoff},
  title   = {A Multi-Scale Analysis of 27,000 Urban Street Networks: Every {US} City, Town, Urbanized Area, and Zillow Neighborhood},
  journal = {Environment and Planning B: Urban Analytics and City Science},
  volume  = {47},
  number  = {4},
  pages   = {590--608},
  year    = {2020},
  doi     = {10.1177/2399808318784595},
  url     = {https://doi.org/10.1177/2399808318784595}
}

@article{CardilloSLP2006,
  author  = {Cardillo, Alessio and Scellato, Salvatore and Latora, Vito and Porta, Sergio},
  title   = {Structural Properties of Planar Graphs of Urban Street Patterns},
  journal = {Physical Review E},
  volume  = {73},
  number  = {6},
  pages   = {066107},
  year    = {2006},
  doi     = {10.1103/PhysRevE.73.066107},
  url     = {https://doi.org/10.1103/PhysRevE.73.066107}
}

@inproceedings{EppsteinGoodrich2008,
  author    = {Eppstein, David and Goodrich, Michael T.},
  title     = {Studying (Non-Planar) Road Networks Through an Algorithmic Lens},
  booktitle = {Proceedings of the 16th ACM SIGSPATIAL International Conference on Advances in Geographic Information Systems},
  pages     = {16:1--16:10},
  publisher = {Association for Computing Machinery},
  year      = {2008},
  doi       = {10.1145/1463434.1463455},
  url       = {https://doi.org/10.1145/1463434.1463455}
}

@book{DowneyFellows2013,
  author    = {Downey, Rodney G. and Fellows, Michael R.},
  title     = {Fundamentals of Parameterized Complexity},
  publisher = {Springer},
  address   = {London},
  year      = {2013},
  doi       = {10.1007/978-1-4471-5559-1},
  url       = {https://doi.org/10.1007/978-1-4471-5559-1}
}

@article{FellowsHRV2009,
  author  = {Fellows, Michael R. and Hermelin, Danny and Rosamond, Frances A. and Vialette, St{\'e}phane},
  title   = {On the Parameterized Complexity of Multiple-Interval Graph Problems},
  journal = {Theoretical Computer Science},
  volume  = {410},
  number  = {1},
  pages   = {53--61},
  year    = {2009},
  doi     = {10.1016/j.tcs.2008.09.065},
  url     = {https://doi.org/10.1016/j.tcs.2008.09.065}
}

\appendix
\section{Appendix}
\label{sec:app}

\subsection{Further Experimental Results}
\label{sec:app:fer}

Recall that for parameters~$(s^*,N^*,m^*,c^*)$ of the TD-separator cut heuristic,
we selected the setup $\chi_1=(14,5000,10,20)$
(cf.~\cref{ssec:algsetup}).
We also compared this setup with two other setups~$\chi_{1/2}=\frac{1}{2}\cdot \chi_1$
and $\chi_{2}=2\cdot \chi_1$,
that is,
when halving and doubling the parameters.
Denote the corresponding ILP-based algorithms
by \ILPcut{1/2} and \ILPcut{2} without preprocessing,
and by \ILPpcut{1/2} and \ILPpcut{2}
when preprocessing is applied.
\cref{fig:solver_performance_tds}
\begin{figure}[t]
 \includegraphics[width=0.3284\textwidth]{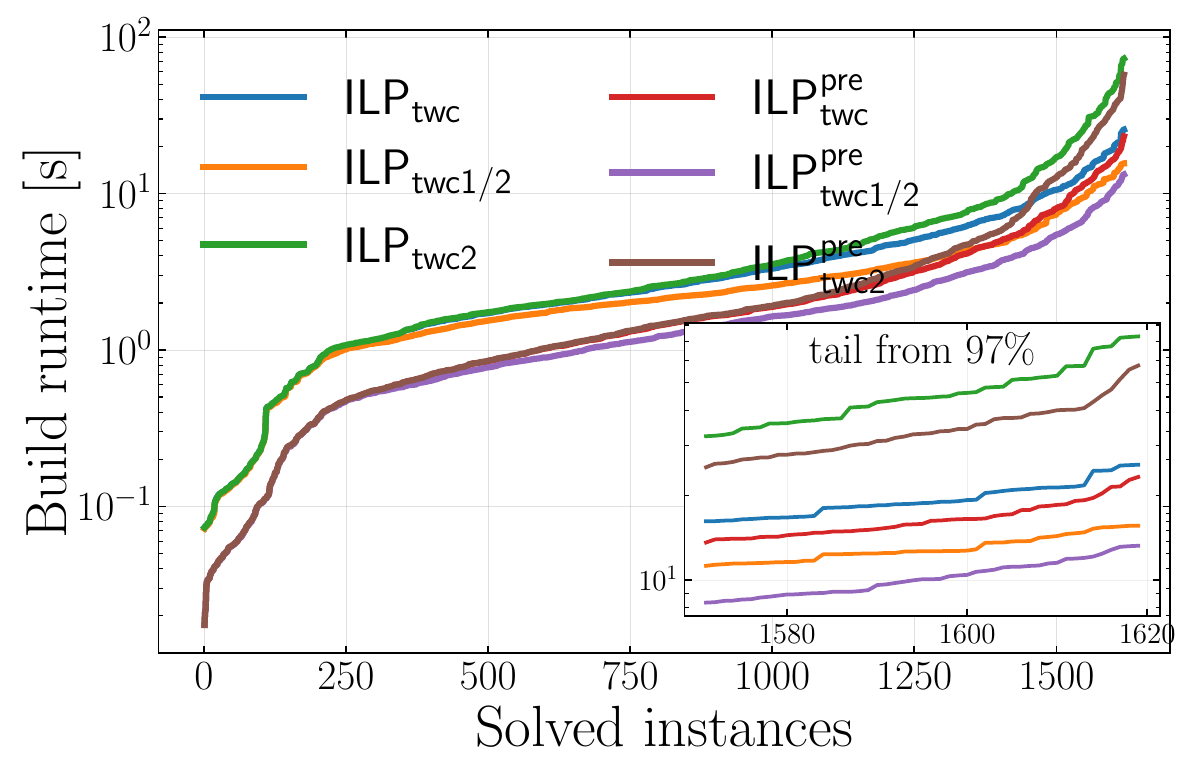}
 \hfill
 \includegraphics[width=0.3284\textwidth]{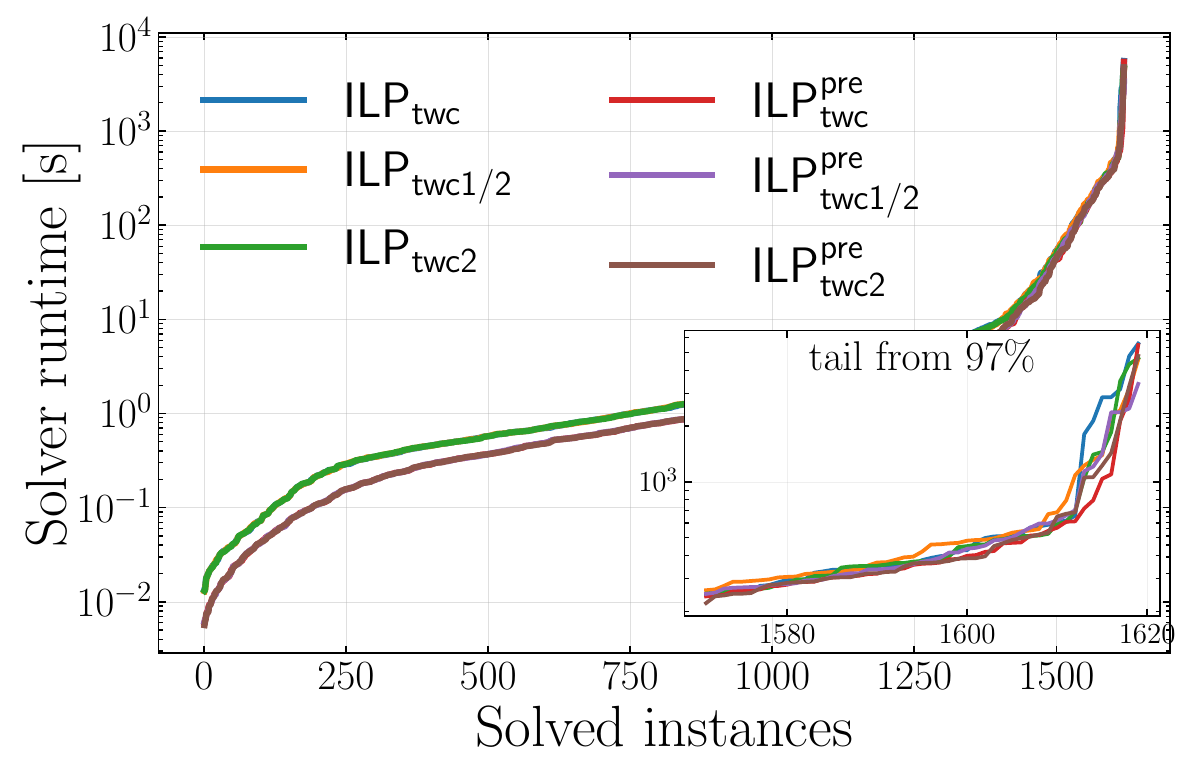}
 \hfill
 \includegraphics[width=0.3284\textwidth]{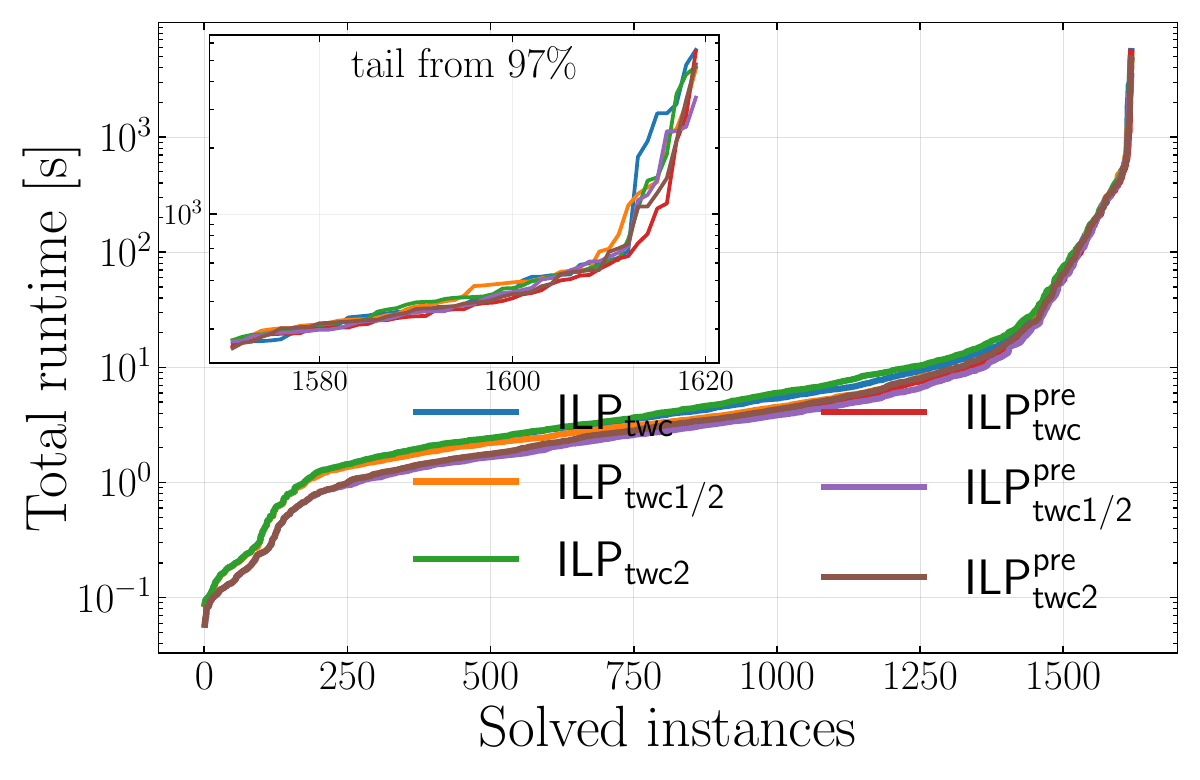}
 \caption{Cactus plots for all region instances solved by each of the three solvers
 based on different parameters with respect to the TD-separator cut heuristic
 with and without preprocessing.
 Presented are building time (left),
 solver runtime (middle),
 and total runtime (right).
 Each inset zooms into the tail of the corresponding cactus plot,
 i.e.,
 the slowest 3\% of the common solved region instances.}
 \label{fig:solver_performance_tds}
\end{figure}
and \cref{fig:scores_tds}
\begin{figure}[t]
 \includegraphics[width=0.3284\textwidth]{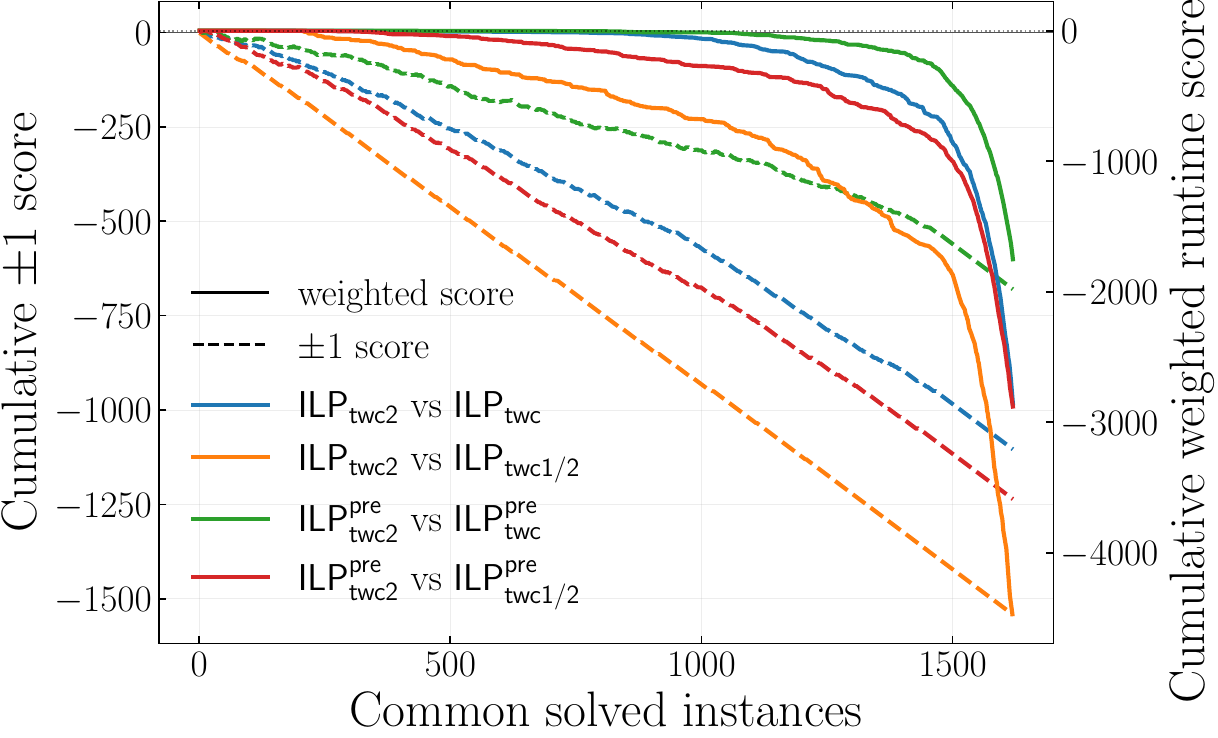}
 \hfill
 \includegraphics[width=0.3284\textwidth]{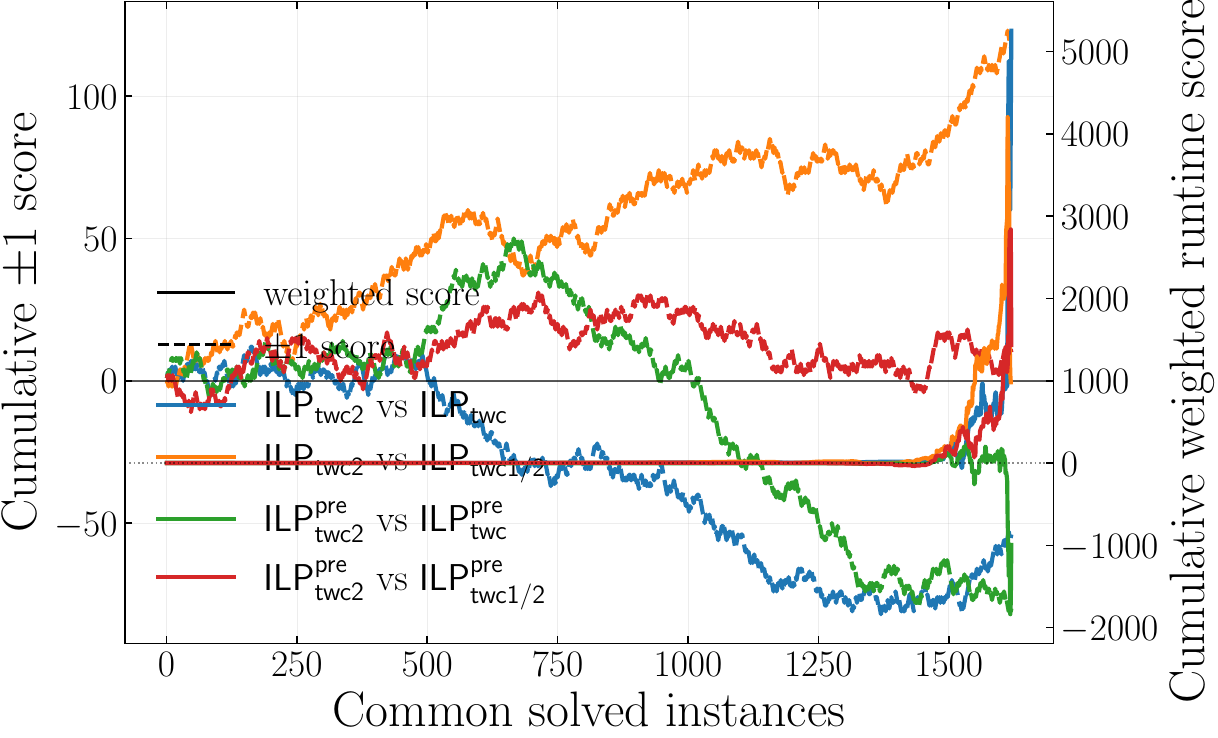}
 \hfill
 \includegraphics[width=0.3284\textwidth]{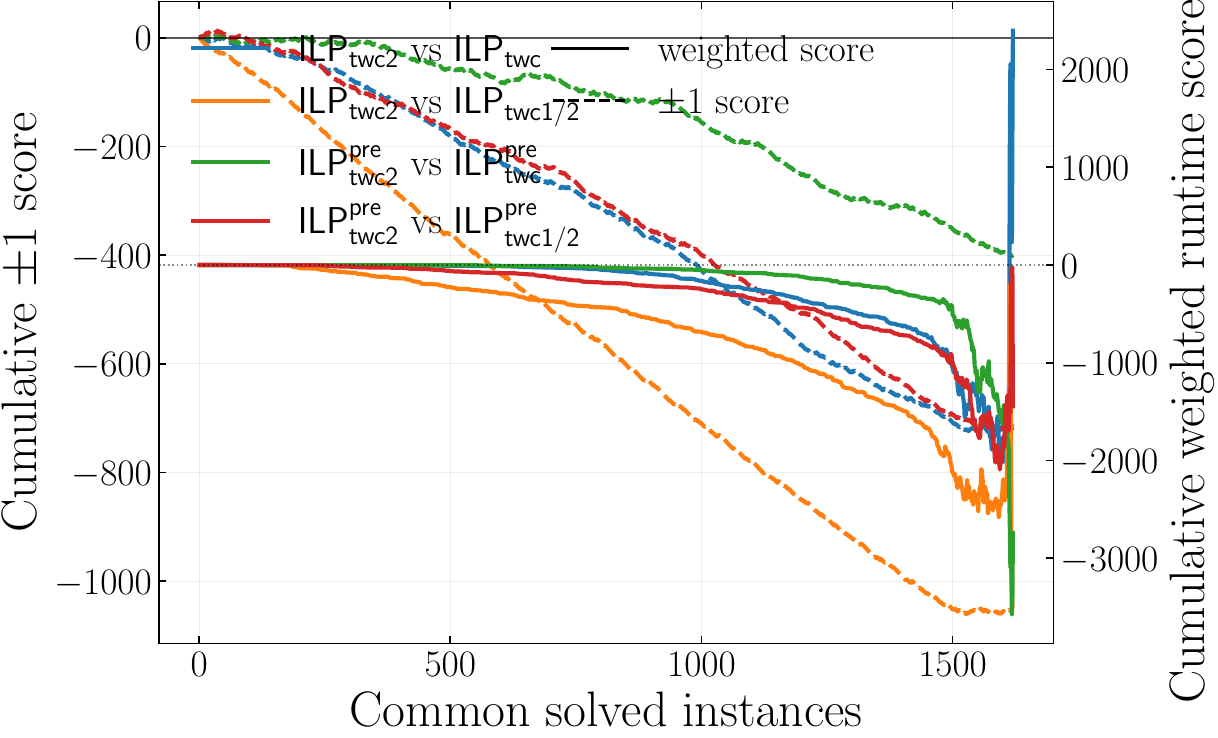}
 \caption{Scores
 (cf.\ \cref{fig:scores})
 of all TD-separator cut heuristics for all commonly solved region instances,
 where the largest setup competes against both the lower two setups,
 with and without preprocessing.
 Presented are building time (left),
 solver runtime (middle),
 and total runtime (right).}
 \label{fig:scores_tds}
\end{figure}
show their performances and
scores on region instances,
respectively.
Larger parameters lead to larger model-building times,
but does not improve performance overall.
Concretely,
\ILPcut{2} has the largest average model-building time
(almost twice the average model-building time of \ILPcut{1/2}),
but the smallest average and median solver runtime
(marginally,
e.g.~regarding the average solver runtime,
\vsec{28.99} compared to \vsec{29.60} for \ILPcut{1/2})
among the commonly solved instances.
In particular,
larger candidate sets interact nontrivially with preprocessing and can even degrade performance.
Concretely,
the average total runtime \ILPpcut{} is the smallest with \vsec{28.92},
yet \ILPpcut{1/2} achieves the smallest median
(\vsec{2.79} compared to second place \vsec{3.07} by \ILPpcut{})
and maximum
(\vsec{3373.82} compared to second place \vsec{4780.89} by \ILPpcut{2})
total runtime
among the commonly solved instances.

\end{document}